\documentclass[aps,prx,10pt,twocolumn,floatfix,superscriptaddress]{revtex4-2}

\usepackage[utf8]{inputenc}
\usepackage{times}
\usepackage{graphicx}
\usepackage{subfigure}
\usepackage{dcolumn}
\usepackage{bm}
\usepackage{xcolor}
\usepackage{amsmath,amssymb}
\usepackage{amsfonts, amsthm}
\usepackage{verbatim}
\usepackage{mathtools}
\usepackage{physics}
\usepackage{epstopdf}
\usepackage{bbold}
\usepackage{verbatim}
\usepackage{hyperref}
\hypersetup{
	colorlinks=true,
	linkcolor=blue,           % color of internal links
	citecolor=blue,           % color of links to bibliography
	filecolor=magenta,        % color of file links
	urlcolor=cyan
}

\newcommand{\diag}{\mathrm{diag}\,}

\newcommand{\R}{\mathbb{R}}

\theoremstyle{definition} 
\theoremstyle{definition} 
\theoremstyle{definition} \newtheorem{prop}{Proposition}[section]
\theoremstyle{definition} \newtheorem{eg}{Example}[section]
\theoremstyle{remark} 
\theoremstyle{definition} \newtheorem{lem}{Lemma}[section]
\theoremstyle{definition} \newtheorem{rem}{Remark}[section]
\theoremstyle{definition} \newtheorem{theo}{Theorem}[section]
\theoremstyle{definition} 
\theoremstyle{definition} 
\theoremstyle{definition} 
\theoremstyle{remark} 

\begin{document}

\title{Training nonlinear optical neural networks with Scattering Backpropagation}

    \author{Nicola Dal Cin}
    \email{nicola.dalcin@mpl.mpg.de}
    \address{Max Planck Institute for the Science of Light, Staudtstraße 2, 91058 Erlangen, Germany}
     \address{Department of Physics, University of Erlangen-Nuremberg, 91058 Erlangen, Germany}
     
    \author{Florian Marquardt}
    \address{Max Planck Institute for the Science of Light, Staudtstraße 2, 91058 Erlangen, Germany}
    \address{Department of Physics, University of Erlangen-Nuremberg, 91058 Erlangen, Germany}

    \author{Clara C. Wanjura}
    \address{Max Planck Institute for the Science of Light, Staudtstraße 2, 91058 Erlangen, Germany}
	
\date{\today}

    
\begin{abstract}
As deep learning applications continue to deploy increasingly large artificial neural networks, the associated high energy demands are creating a need for alternative neuromorphic approaches. Optics and photonics are particularly compelling platforms as they offer high speeds and energy efficiency. Neuromorphic systems based on nonlinear optics promise high expressivity with a minimal number of parameters. However, so far, there is no efficient and generic physics-based training method allowing us to extract gradients for the most general class of nonlinear optical systems. In this work, we present Scattering Backpropagation, an efficient method for experimentally measuring approximated gradients for nonlinear optical neural networks. Remarkably, our approach does not require a mathematical model of the physical nonlinearity, and only involves two scattering experiments to extract all gradient approximations. The estimation precision depends on the deviation from reciprocity. 
We successfully apply our method to well-known benchmarks such as XOR and MNIST.
Scattering Backpropagation is widely applicable to existing state-of-the-art, scalable platforms, such as optics, microwave, and also extends to other physical platforms such as electrical circuits.
\end{abstract}

\maketitle

\section{Introduction}

\begin{figure}[h!]
    \centering
    \includegraphics[width=\linewidth]{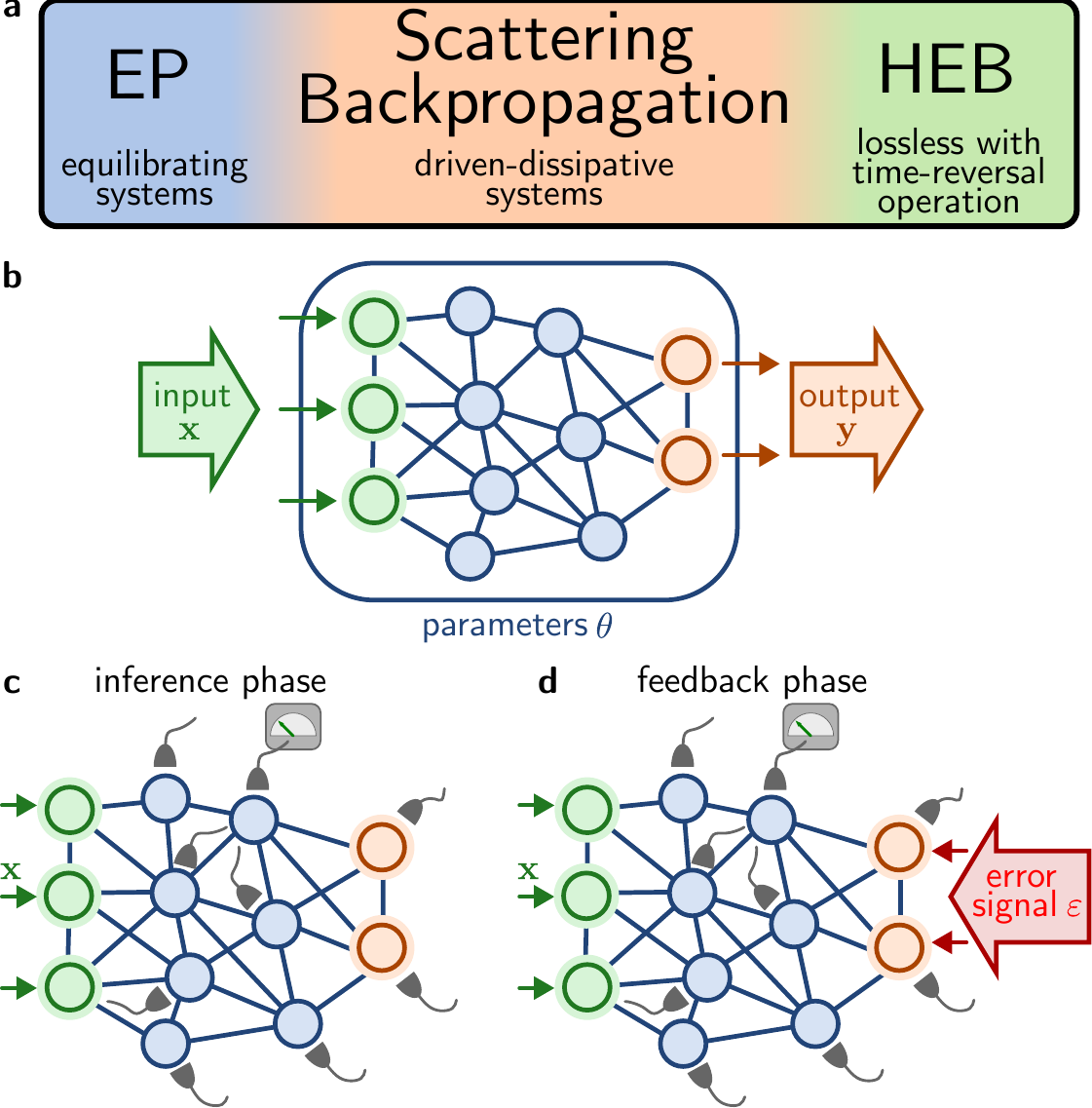}
    \caption{\textbf{Training of a nonlinear, neuromorphic network of resonators.}
    \textbf{a)}~Overview of generic, neuromorphic training methods.
    \textbf{b)}~Schematic representation of an optical neuromorphic system described by the dynamical equations  \eqref{eqn:nonl-system}, in which each mode $a_j$ corresponds to a node in the network (e.g. modes in coupled optical resonators). The input $\mathbf{x}$ is encoded in the input light fields incident on a suitable subset of resonators, while the output $\mathbf{y}$ is readout as the light exiting from another subset of resonators.
    Nonlinear data processing is enabled by a nonlinear physical process; e.g., the resonators could be Kerr-resonators, or the coupling between resonator modes could be via cross-Kerr coupling. Tunable, physical parameters such as resonator frequencies and the couplings between them, collected in the set $\theta$, are adjusted during training.
    We propose an efficient physics-based method for training the resonator network which involves two steps.
    \textbf{c)}~In the inference phase, we inject the input $\mathbf{x}$ and measure the output light field at each resonator connected to a trainable parameter. This can be performed simultaneously for all resonators. \textbf{d)}~In the feedback phase, we inject the error signal, computed from the cost function as a measure how far the output deviates from the target output, and again measure the output at all resonators connected to trainable parameters.
    The necessary gradients for the parameter update are computed from the output signals in the two phases, Eqs.~\eqref{eq:gradientFreuquencies} and \eqref{eq:gradientCouplings}.
    }
    \label{fig:setup}
\end{figure}%
The growing energy demand of machine learning and artificial intelligence has created a need for alternative approaches in the form of neuromorphic hardware~\cite{markovic2020physics}, relying on analogue physical neural networks promising high computation speeds at low energy consumption.
There is a variety of suitable neuromorphic computing platforms~\cite{wetzstein2020inference,shastri2021photonics,grollier2020neuromorphic,schneider2022supermind} 
promising efficient computation
of which optical systems~\cite{wagner1987multilayer,shastri2021photonics,wetzstein2020inference} are particularly appealing as they promise highly parallel linear computations, which is one important aspect of a neural network
operating at high speeds and bandwidth.
In general, neural networks perform nonlinear computations on the data. For optical systems, this can be achieved either using nonlinear optical elements~\cite{Zuo2019,Basani2024All}, through hybrid approaches involving an optoelectronic conversion step~\cite{hughes2018training,hamerly2019large,pai2023experimentally,chen2023deep},
or with a purely linear optical scattering system by encoding the input data in the system parameters~\cite{Wanjura2024Fully,Yildirim2024Nonlinear,Xia2024Nonlinear}.

Even though we have a plurality of possible neuromorphic platforms to choose from, \emph{training} them is still very much an open challenge for many setups~\cite{momeni2024training}.
%Ideally, we would like to be able to extract 
While digital neural networks are trained with the backpropagation algorithm~\cite{rumelhart1986learning}, we have to obtain the gradients needed for training physical neural networks by some other means.
In-silico training is often unsuccessful since discrepancies between the simulated model and the real physical system lead to an accumulation of errors during the simulation. Physics-aware training~\cite{wright2022deep} only simulates the backwards pass but utilises the physical forward pass and has been shown to perform more reliably, but still requires a faithful digital model.
The conceptually simplest approach which can be applied to any physical neural network is the parameter shift method~\cite{Filipovich2022Silicon,bandyopadhyay2024single}. Here, parameters are shifted one by one to approximate the gradients. However, this scales unfavourably with network size~\cite{bartunov2018assessing}.
In a pioneering work~\cite{wagner1987multilayer}, Psaltis et al. formulated the first physics-based training method which extracts the gradients needed for training in an optical network based on volume holograms~\cite{psaltis1990holography}. The nonlinear optical element has to be carefully designed so that the transmittance in the backward direction realises the gradient of the transmittance in the forward direction.
Even though, this scheme was later implemented to some degree~\cite{li1993optical,wagner1993optical,skinner1995neural}, the strict requirements and practical limitations of this scheme imply that it cannot be applied to generic optical and nonlinear systems.
Other training approaches only perform backpropagation on the linear components~\cite{hughes2018training,pai2023experimentally} or are tailored to specific nonlinearities~\cite{guo2021backpropagation,spall2023training}.
While a few attempts have been made to train gradient-free~\cite{Momeni2023Backpropagation}, %forward-forward
still the target of most approaches is to extract the gradients.

There are only two generic physics-based training methods applying to different limiting cases, Fig.~\ref{fig:setup}~\textbf{a}. Equilibrium Propagation~\cite{scellier2017equilibrium,martin2021eqspike,stern2021supervised} applies to equilibrating systems whose dynamics are determined by an energy function, while Hamiltonian Echo Backpropagation~\cite{lopez2021self} can be deployed for the training of lossless systems with a time-reversal operation, such as phase conjugation in optical systems.

A large class of nonlinear systems is actually not covered by these training methods since those systems are both out of equilibrium and have losses. In fact, for the default neuromorphic system in optics, namely a nonlinear, driven-dissipative optical system, there is no efficient physics-based training method, yet.
The lack of efficient physics-based training strategies for nonlinear optical systems with dissipation is one of the reasons why these systems have so far not been realised experimentally at scale.

In this work, we fill this gap and develop a physics-based training method for this large class of nonlinear optical or photonic systems. Our method approximates the needed gradients by comparing only two scattering response experiments, ensuring efficient gradient extraction independent of the number of training parameters. As one significant advantage of our method, it can be applied to any nonlinear optical system and does not require a faithful model for the physical nonlinearity. The only requirements are (i)~that the trainable parameters enter in the linear part of the physical system, (ii) a stable steady state exists, and (iii)~the forward and backward scattering response through the system are approximately the same (reciprocity). 
We demonstrate successful training on typical nonlinear benchmark tasks such as XOR and MNIST for simulated systems of coupled resonators with Kerr- or cross-Kerr nonlinearity.
Beyond optical systems, our method applies more generally to a large class of nonlinear physical systems, e.g. to many port-Hamiltonian systems.

Our work opens the door to efficient neuromorphic computing with optical and photonic systems for which previously no efficient training method existed.

\section{Setup}
We consider a general, driven, nonlinear optical system of $N$ coupled modes $a_j$, e.g. hosted in optical resonators, as shown in Fig.~\ref{fig:setup}~\textbf{b}. 
The system's time evolution can be described in terms of dynamical equations, which, collecting the modes $a_j$ into a vector $\bm{a}=(a_1,\dots,a_N)^\mathsf{T}$, take the form
\begin{equation}
\label{eqn:nonl-system}
    \dot{\bm{a}}(t) = - i H(\theta) \bm{a}(t) - i g \bm\varphi(\bm{a}(t)) - \sqrt{\kappa}\, \bm{a}_\mathrm{in}(\mathbf{x}).
\end{equation}
The first term on the right-hand side of Eq.~\eqref{eqn:nonl-system} describes the system evolution due to linear physical interactions, where we also encode the training parameters $\theta$.
To be specific, the dynamic matrix $H$ can encode the mode frequency detunings $\Delta_j$, the decay rates due to intrinsic losses $\kappa_j'$ and external ones $\kappa_j$ due to the coupling with an external bath (e.g. any potential probe waveguide), $H_{j,j}=-i(\kappa_j+\kappa'_j)/2+\Delta_j$, as well as the real symmetric couplings between modes, $H_{j,\ell}=J_{j,\ell}$ for $j \neq \ell$.
The second term in Eq.~\eqref{eqn:nonl-system} denotes a generic physical nonlinearity $\bm\varphi$ whose strength is controlled by the factor $g$. Notably, we do not need to know the specific type of nonlinearity to apply our training method. For concreteness, we will later consider a network of coupled resonators with self-Kerr and cross-Kerr nonlinearities as examples. With $\bm{a}_\mathrm{in}=(a_{\mathrm{in},1},\dots,a_{\mathrm{in},N})^\mathsf{T}$ in Eq.~\eqref{eqn:nonl-system} we indicate the vector of probe fields $a_{\mathrm{in},j}$ injected at the $j$th mode. $\kappa$ denotes a diagonal matrix of all the decay rates $\kappa_j$.

Perhaps the most widespread and obvious technique to encode the input in an optics-based physical neural network is via the amplitudes of the light. Accordingly,
%In accordance to the popular choice of encoding the input $\mathbf{x}$ into the amplitudes of the light in many optical physical neural networks,
we encode $\mathrm{x}_j$ into a suitable set of input modes $a_{\mathrm{in},j} / \sqrt{\bar{{\kappa}}}$ with $j \in \mathcal{I}_\mathrm{in}$ and a suitable reference rate $\bar\kappa$, e.g., the average loss rate.
The resulting system response is encoded in the output fields $a_{\mathrm{out},j}$ connected to the input fields $a_{\mathrm{in},j}$ and the fields $a_j$ within each mode according to the input-output relations $\bm{a}_\mathrm{out}(t) = \bm{a}_\mathrm{in} + \sqrt{\kappa} \, \bm{a}(t)$~\cite{Gardiner1985Input,Clerk2010Introduction}.
The neuromorphic system's output $\mathbf{y}$ is then given by the output field $a_{\mathrm{out},j}/\sqrt{\bar\kappa}$ at a suitable set of modes $j \in \mathcal{I}_\mathrm{out}$.

During supervised training, we minimize a cost function $C(\mathbf{y}, \mathbf{y}_\mathrm{target})$ which quantifies the deviation between the network output $\mathbf{y}$ and the target output $\mathbf{y}_\mathrm{target}$ for a given input $\mathbf{x}$.
At each training step, we update the trainable physical parameters $\theta$ (e.g. detunings $\Delta_j$ and couplings $J_{j,\ell}$), according to
\begin{equation}
\label{eqn:weight-update}
    \theta \leftarrow \theta - \eta \frac{\partial C}{\partial \theta}(\mathbf{y}, \mathbf{y}_\mathrm{target}), \quad \eta >0.
\end{equation}
In the following, we introduce Scattering Backpropagation, an efficient method for physically extracting an approximation of the gradient $\frac{\partial C}{\partial \theta}(\mathbf{y}, \mathbf{y}_\mathrm{target})$ with a minimal number of scattering experiments.

\section{Outline of scattering backpropagation}
\label{sec:extracting_gradients}

Before providing a more general formulation of our training method, we summarize the key steps for applying Scattering Backpropagation to nonlinear optical systems. Mathematical details are provided in Methods and Supplementary Information (SI).
Our proposed gradient extraction procedure consists of two phases. (i)~In the \textit{inference phase}, Fig.~\ref{fig:setup}~\textbf{c}, a probe signal $\bm{a}_\mathrm{in}$, encoding the network input $\mathbf{x}$, is injected into the system and, after it reaches a steady state $\bar{\bm{a}},$ we measure the response field $\bm{a}_\mathrm{out}$ at every mode. %node $1\leq j \leq N$.
In particular, we use the measured output $\mathbf{y}$ to compute
the loss $C(\mathbf{y},\mathbf{y}_\mathrm{target})$.
(ii)~In the \textit{feedback phase}, Fig.~\ref{fig:setup}~\textbf{d}, 
we compute the \textit{error signal} $\frac{\partial C}{\partial \mathbf{y}}$, feed it back to the system as detailed below
and observe the system's reaction which informs us about the gradients $\partial_\theta C$.
Concretely, we add the error signal as a small perturbation %(or \textit{error signal})
$\delta \bm{a}_\mathrm{in}$ to a probe field $\bm{a}_\mathrm{in}$ incident on the output nodes (while keeping the input field on the input nodes fixed) which brings the output of the system closer to the target $\mathbf{y}_\mathrm{target}$.
Specifically,
\begin{equation}
\label{eq:error_signal}
    \delta \bm{a}_\mathrm{in} \coloneqq -i \beta \frac{\partial C}{\partial \bm{a}_\mathrm{out}}(\mathbf{y},\mathbf{y}_\mathrm{target}),
\end{equation}
in which $\partial$ denotes the Wirtinger derivative \cite{remmert1991theory}. Note that the non-zero entries of $\frac{\partial C}{\partial \bm{a}_\mathrm{out}}(\mathbf{y},\mathbf{y}_\mathrm{target})$ correspond to the error signal $\frac{\partial C}{\partial \mathbf{y}}$, and $\beta$ is in units of a loss rate.
Adapting to this new input, the network evolves into a new but close steady state $\bar{\bm{a}} + \delta \bar{\bm{a}}$, producing a new output field $\bm{a}_\mathrm{out} + \delta \bm{a}_\mathrm{out}$ which is again measured at every node. Here, $\bm{a}_\mathrm{out}$ was the response we recorded in the inference phase. We call the additional signal $\delta \bm{a}_\mathrm{out}$ the \emph{learning response}.

Given the response in the free phase at every node $\bm a_\mathrm{out}$ and the learning response $\delta \bm{a}_\mathrm{out}$, we can now compute the gradients w.r.t. the training parameters.
Concretely, for the gradients w.r.t. detunings $\Delta_j$ and couplings $J_{j,\ell}$, we obtain
\begin{equation}\label{eq:gradientFreuquencies}
   \frac{\partial C}{\partial \Delta_j} \approx  -\frac{2}{\kappa_j} \mathfrak{Re}\left[ (\bm a_{\mathrm{out},j}-\bm a_{\mathrm{in},j})\frac{\delta \bm a_{\mathrm{out},j}- \delta \bm a_{\mathrm{in},j}}{\beta} \right],
\end{equation}
and
\begin{align}\label{eq:gradientCouplings}
\notag
     \frac{\partial C}{\partial J_{j,\ell}} \approx - \frac{2}{\sqrt{\kappa_j \kappa_\ell}}& \mathfrak{Re}\biggl[ (\bm a_{\mathrm{out},\ell}-\bm a_{\mathrm{in},\ell})\frac{\delta \bm a_{\mathrm{out},j}- \delta \bm a_{\mathrm{in},j}}{\beta}\\
     &+ (\bm a_{\mathrm{out},j}-\bm a_{\mathrm{in},j})\frac{\delta \bm a_{\mathrm{out},\ell}- \delta \bm a_{\mathrm{in},\ell}}{\beta} \biggr].
\end{align}
The level of the approximation above depends on the deviation of the system response from reciprocity (next section and Methods), i.e. the breaking of a symmetry of the scattering matrix.
Remarkably, since we can measure all the fields $\bm{a}_{\mathrm{out},j}$, $\delta \bm{a}_{\mathrm{out},j}$ in parallel, we can measure all gradients with only two measurements.
Furthermore, we see that the gradients in  Eqs.~\eqref{eq:gradientFreuquencies} and \eqref{eq:gradientCouplings} only depend on locally measured quantities which do not require full knowledge of the system. In particular, it is not necessary to know the type of nonlinearity. The only requirement is that all tunable parameters enter in the linear contribution in Eq.~\eqref{eqn:nonl-system}. Hence, Scattering Backpropagation can be applied to so called \emph{grey box} systems, i.e., systems for which one part is known while other parts can be unknown.

\section{General formulation of the method}
\label{sec:approx_reciprocity}

\subsection{Dynamical systems with input-output relations}
The training method outlined above for optical systems is in fact more general and can be formulated for a large class of parametrized, driven, autonomous dynamical systems
\begin{equation}
\label{eq:dyn_sys}
\dot{\bm{\xi}}=\bm{F}_\theta(\bm{\xi})-\sqrt{\kappa} \, \bm{\xi}_\mathrm{in},
\end{equation}
which also includes our dynamical equations  \eqref{eqn:nonl-system} for $\bm{\xi}=(\bm{a},\bm{a}^*)$, and $\bm{\xi}_\mathrm{in}=(\bm{a}_\mathrm{in},\bm{a}^*_\mathrm{in})$, and more generally some much studied ``port Hamiltonian" systems \cite{van2014port}.
Any system described by equations of the form~\eqref{eq:dyn_sys} together with linear input-output relations, e.g. $\bm\xi_\mathrm{out}=\bm\xi_\mathrm{in}+\sqrt{\kappa} \, \bm{\xi}$, can be trained with our method (SI).

For this class of systems $\partial_\theta C(\mathbf{y},\mathbf{y}_\mathrm{target})$ depends on the \emph{linearized scattering matrix} $S_\theta(\bar{\bm\xi})$ (Methods), which is defined via the Green's function at the steady state, and determines the linear response of the system to a small input perturbation $\delta \bm{\xi}_\mathrm{in}$ on top of the original signal $\bm{\xi}_\mathrm{in}$:
\begin{equation}
\label{eq:in-out-relations}
 \delta \bm\xi_\mathrm{out} = S_\theta(\bar{\bm\xi}) \, \delta \bm\xi_\mathrm{in} + \mathcal{O}(\delta \bm\xi_\mathrm{in}^2).
\end{equation}
While, in principle, Eq.~\eqref{eq:in-out-relations} lets us physically extract all gradients, this would require a number of measurements scaling with system size. To formulate an efficient extraction method, we approximate the gradients by utilizing some (approximate) symmetry of the system, which
in optical systems is given by reciprocity. In general,
\begin{align}
  \label{eqn:quasi-symmetry-1}  S_\theta(\bar{\bm\xi})^\dagger = U S_\theta(\bar{\bm\xi})U^{-1}+ \mathcal{O}(g),
\end{align}
where $U$ is an invertible matrix denoting the symmetry.
For efficient training, $U$ should be a local transformation.
In optical systems,  $g$ is the nonlinearity strength.
From Eq.~\eqref{eqn:quasi-symmetry-1} and the error signal
\begin{equation}
    \delta \bm{\xi}_\mathrm{in}\coloneqq \beta U^{-1}\frac{\partial C}{\partial \bm\xi^*_\mathrm{out}}(\mathbf{y},\mathbf{y}_\mathrm{target}),
\end{equation}
we obtain the gradient approximation (Methods)
\begin{align}
\label{eq:general_gradient_formula}
    \frac{\partial C}{\partial \theta}
        &= - \left(\frac{\partial \bm F_\theta}{\partial  \theta }(\bar{\bm\xi})\right)^\dagger \sqrt{\kappa^{-1}} U \; \frac{\delta \bm \xi_\mathrm{out} - \delta \bm \xi_\mathrm{in}}{\beta} + \mathcal{O}(g,\beta).
\end{align}
We note that $\bm{\xi}_\mathrm{out}$ and $\delta \bm{\xi}_\mathrm{out}$ only involve local measurements during inference and feedback phase.
In the optical case, in which the trainable parameters only enter linearly in Eq.~\eqref{eqn:nonl-system}, Eq.~\eqref{eq:general_gradient_formula} produces Eqs.~\eqref{eq:gradientFreuquencies} and \eqref{eq:gradientCouplings}. If the trainable parameters enter in the nonlinear part in Eq.~\eqref{eqn:nonl-system}, the form of the nonlinearity needs to be known to evaluate $\partial \bm F_\theta (\bar{\bm\xi})/\partial \theta$ in Eq.~\eqref{eq:general_gradient_formula}.

\subsection{Application to optical systems and quasi-reciprocity}

In this section, we introduce (approximate) symmetries $U$ for optical systems.
A linear optical system with real and symmetric couplings $J_{j, \ell}$ in Eq.~\eqref{eqn:nonl-system} with $g=0$ is reciprocal, $S_{j,\ell} = S_{\ell,j}$.
This is equivalent to $S_\theta^\dagger = U S_\theta U^{-1}$
with either $U=\sigma_x$ or $U=\sigma_y$, in which
$\sigma_x=\begin{pmatrix}
    0 & \mathbf{I}_N\\
    \mathbf{I}_N & 0
\end{pmatrix}$, and $\sigma_y=\begin{pmatrix}
    0 & -i\mathbf{I}_N\\
    i\mathbf{I}_N & 0
\end{pmatrix}$. Notice that $\sigma_x$ is a \emph{local} transformation since it exchanges every mode $a_j$ with its conjugate $a_j^*$, while $\sigma_y$ does the same after applying a phase shift.

It is well known~\cite{caloz2018electromagnetic}, that nonlinearities in optics can break reciprocity:  this is also the case for our optical system~\eqref{eqn:nonl-system} for $g\neq0$. In particular, 
the nonlinearity introduces a coupling between $\delta \bar{\bm{a}}$ and $\delta \bar{\bm{a}}^*$, due to the non-holomorphic nature of the optical nonlinearity, which breaks the symmetry above.
However, depending on the nonlinearity strength $g$, the dissipation, and input intensity, it is still possible to observe approximate reciprocity, a regime we call \emph{quasi-reciprocity} (SI).
Indeed, in this case we obtain Eq.~\eqref{eqn:quasi-symmetry-1} for the linearized scattering matrix $S_\theta(\bar{\bm{a}},\bar{\bm{a}}^*)$.
Since the reciprocity is only weakly broken, the error signal $\delta \bm{a}_\mathrm{in}$ injected at the output sites in the feedback-phase, carries back almost the same information about how light scattered in the `forward pass' during the inference phase. Therefore, this approximation is not an obstacle for the efficient gradient extraction. We empirically quantify this approximation in further detail below.

\begin{figure*}
    \centering
\includegraphics[width=\linewidth]{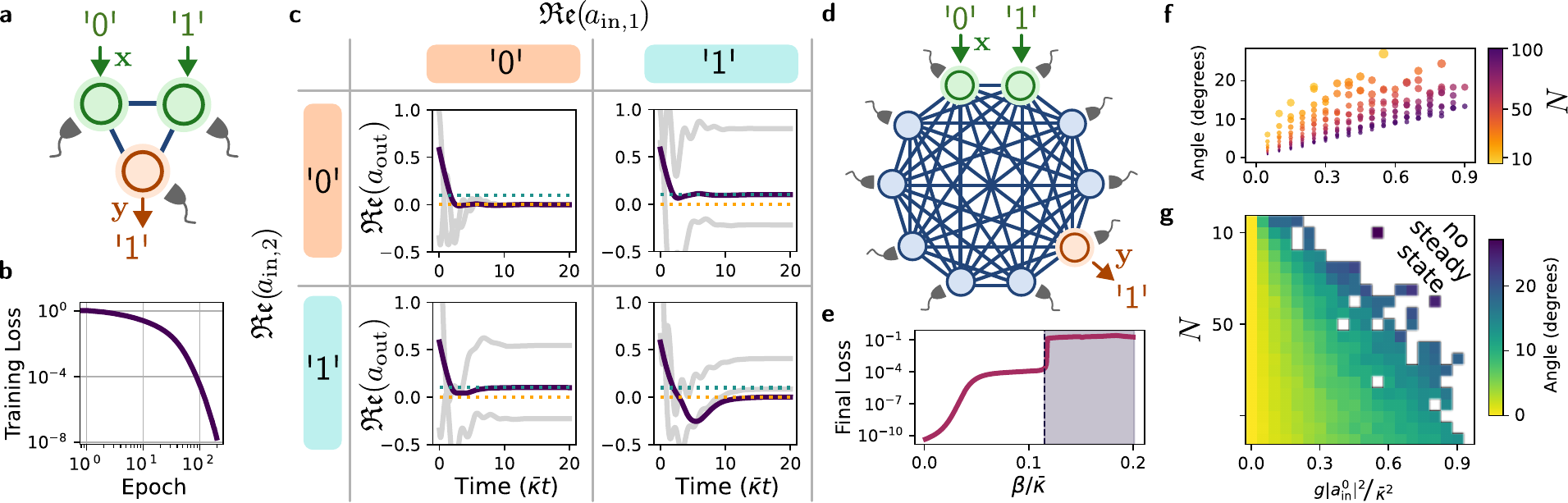}
    \caption{\textbf{Scattering Backpropagation Training} -- 
    \textbf{a)}~ A fully connected optical network with $N=3$ nodes and self-Kerr nonlinearities of strength  $g/\bar{\kappa}=0.2$, $|a_\mathrm{in}^0|/\sqrt{\bar{\kappa}}=1$ can already learn XOR. \textbf{b)}~ The mean square error loss evolution during the training of the system. \textbf{c)}~ Time evolution of the three output amplitudes $\bm{a}_\mathrm{out}$ during inference, after switching the input signals to the indicated configurations (blue: network output), after successful training. The dotted lines correspond to the target values for  logical $0$ and $1$.
    \textbf{d)}~ Neuromorphic system with $N=10$ nodes linearly coupled all-to-all; the learnable parameters are the detunings $\Delta_j$ and the couplings $J_{j,\ell}$. In our simulations, we assume $\kappa_j = \bar{\kappa}$ for all $j$. \textbf{e)} Dependence of training success on the scale $\beta$ of the perturbation during gradient extraction. The final loss after training for $1000$ epochs remains small for rather large $\beta$, outside of the gray area [$N=10$, $g/\bar{\kappa}=0.2$, $|a_\mathrm{in}^0|/\sqrt{\bar{\kappa}}=1$].  \textbf{f)} and \textbf{g)} Angle between $S_\theta(\bar{\bm{a}}, \bar{\bm{a}}^*)^\dagger$ and $\sigma_y S_\theta(\bar{\bm{a}}, \bar{\bm{a}}^*)\sigma_y$ (related to the angle between the true gradient $\partial_\theta C$ and the approximation used in Scattering Backpropagation) for a system of $N$ modes, linearly coupled all-to-all, with self-Kerr nonlinearities. Here we scale the nonlinearity in the suitable form  $g |a^{0}_\mathrm{in}|^2/\bar{\kappa}^2$, with $|a^{0}_\mathrm{in}|$ a reference input amplitude. Each data point represents the angle averaged over $50$ simulations, removing runs that do not converge to a steady state.
    The gradient approximation remains good throughout almost the entire stable regime of the nonlinear system.
    Both the average of the angles and their standard deviation (represented by the radii of the dots in the plot) are proportional to $g$ and $1/N.$ 
    \label{fig:XOR} }
\end{figure*}

\subsection{Connection to Equilibrium Propagation for vector fields}

Our approach promises efficient physics-based gradient extraction in nonequilibrium driven-dissipative coupled-mode systems (e.g. in optics). The closest existing approach would be a generalization to vector field dynamics  \cite{scellier2018generalization} of Equilibrium Propagation
\cite{scellier2017equilibrium}, in which the accuracy of the approximation depends on the symmetry of the weights. However, it turns out that approach is not directly applicable to the class of systems Eq.~\eqref{eqn:nonl-system} considered in our work, which include many examples of the much-studied and widely applied port Hamiltonian systems \cite{van2014port} (SI). Another difference from our proposed approach is that during the nudged phase Equilibrium Propagation requires engineering a change in the dynamical equations that depends on the cost function.

It is however possible to view our training method, in the particular case $\bm{a}_\mathrm{out} = \bm{a}(t)$, as a generalized version of Equilibrium Propagation for vector fields where we allow for more general (quasi)-symmetries in the Jacobian and we evaluate $\frac{\partial C}{\partial \bm{a}}(\bm{a}, \mathbf{y}_\mathrm{target})$ at the steady state $\bar{\bm{a}}$ (SI).

\section{Case studies}
\subsection{Network of Kerr resonators}

To test our general training approach, we simulate the training of a neuromorphic scattering system described by Eq.~\eqref{eqn:nonl-system}. As an example, see Fig.~\ref{fig:XOR}~\textbf{a}, we consider a network of $N$ resonators in which the trainable weights are the symmetric couplings $J_{j,\ell}$ and the detunings $\Delta_j \coloneqq J_{j,j}$.  We assume the system to feature self-Kerr nonlinearities, i.e. $\varphi_j(\bm{a})\coloneqq |a_j|^2 a_j$. In the SI, we also present examples of training in systems with cross Kerr-nonlinearities.

For a given input $\mathbf{x}$, encoded in the input field $\bm{a}_\mathrm{in},$ we solve the dynamical equations up to some $t_{\rm max}$, starting from random initial conditions. The output $\mathbf{y}$ of the neuromorphic system is then defined in terms of the field $ \bm{a}_\mathrm{out} = \bm{a}_\mathrm{in} + \sqrt{\kappa}\bm{a}(t_{\rm max}) \approx \bm{a}_\mathrm{in} + \sqrt{\kappa}\bar{\bm{a}}.$
After this inference phase, we compute the error signal $\delta \bm{a}_\mathrm{in}$ and solve the perturbed dynamics starting from the old steady state $\bm{a}(t_{\rm max})$. Finally, we use $\bm{a}_\mathrm{out}$, $\delta \bm{a}_\mathrm{in}$, and  $\delta \bm{a}_\mathrm{out}$ to compute the approximate gradients Eq.~\eqref{eq:gradientFreuquencies} and Eq.~\eqref{eq:gradientCouplings} and update the system's parameters.

\subsection{Training a small network}

We first demonstrate how our method can be applied to train a small system of $N=3$ coupled Kerr-resonators to learn XOR (Fig.~\ref{fig:XOR}~\textbf{a}). The network input is encoded in the real parts of the input fields $a_{\mathrm{in},1}$ and $a_{\mathrm{in},2}$, while the real part of $a_{\rm \mathrm{out},3}$ represents the output. The loss decreases monotonically during training (Fig.~\ref{fig:XOR}~\textbf{b})  even though the physically extractable gradient that we employ is only approximate and we do not follow directly the steepest descent in the loss landscape. In Fig.~\ref{fig:XOR}~\textbf{c}, we plot the time evolution of the $\bm{a}_\mathrm{out}(t)$ modes for the successfully trained model and observe that the steady state reached in this nonlinear system is independent of initial conditions (SI). 

\subsection{Analysis of approximations}

One significant aspect of our general training method for scattering systems is the approximate nature of the physically extracted gradient.
Therefore, it is important to analyze whether the angle between the true gradient  $\partial_\theta C$  and the approximation (Eqs.~\eqref{eq:gradientFreuquencies} and \eqref{eq:gradientCouplings}) remains small, to guarantee successful training. It is possible to prove (SI) that this angle is related to the angle $\alpha$ between $A=S_\theta(\bar{\bm{a}},\bar{\bm{a}}^*)^\dagger$ and $B=\sigma_y S_\theta(\bar{\bm{a}},\bar{\bm{a}}^*) \sigma_y$. This latter angle also determines the precision of the approximation~\eqref{eq:general_gradient_formula}, and it depends only on the system, not on the cost function of the specific task. It can be defined by fixing any inner product, e.g. the Frobenius $\langle A,B\rangle_F \coloneqq \Tr(A^\dagger B)$, with $\langle A,B\rangle_F = \|  A \|_F \| B \|_F \cos\alpha $. In Fig.~\ref{fig:XOR}~\textbf{f, g}  we analyze the behavior of the angle $\alpha$ as a function of  system size $N$ and nonlinearity for randomly sampled systems. Qualitatively, we observe a linear dependence between $\alpha$ and the strength $g$ of the nonlinearity, matching the analytical results (SI).
Interestingly, we also observe that for ``sparse" nonlinearity, like onsite self-Kerr nonlinearity or cross-Kerr nonlinearities connecting the modes in a line/circle (SI), the approximation improves for larger systems (with a $1/N$ behavior seen in Fig.~\ref{fig:XOR}~\textbf{f} ). 

In summary, the angle remains small even for fairly large nonlinearities, almost up to the threshold where the system does not reach a steady state anymore. Therefore, we can obtain both good gradient approximations and nonlinear expressivity simultaneously, especially when we scale up the system.

On another note, in any real optical setup, the measurements extracting the training gradients from scattering experiments will display shot noise. Therefore, one needs to work at finite values of the parameter $\beta$ that determines the strength of the perturbation that is injected, to allow for good signal/noise ratio. We have analyzed how far $\beta$ can be pushed while still allowing for successful training (Fig.~\ref{fig:XOR}~\textbf{e}). We observe that training works even for values much larger than $10^{-2}$, the value we employed for the other simulations.

\subsection{Image recognition}
\begin{figure}
    \centering
    \includegraphics[width=\linewidth]{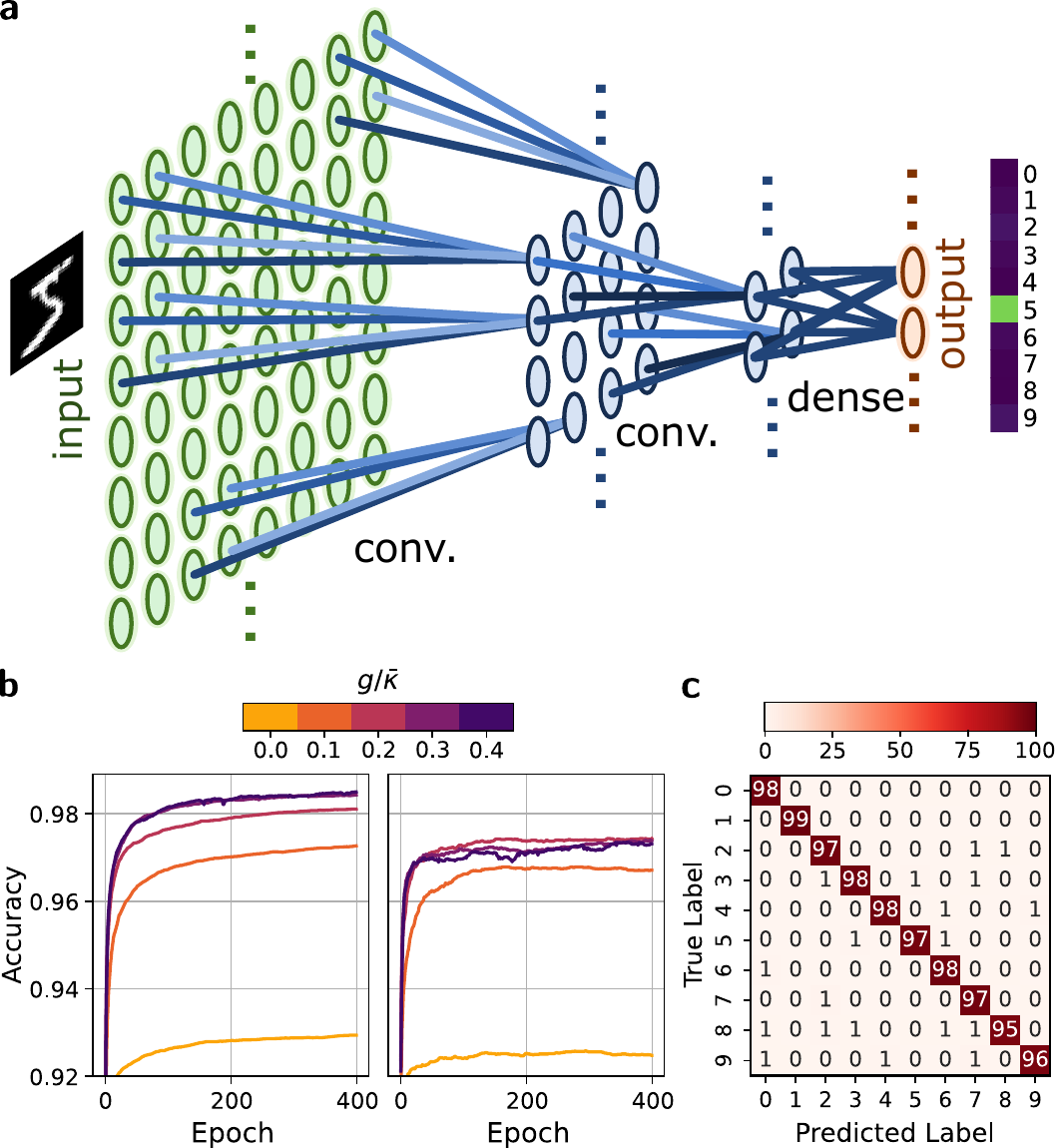}
\caption{\textbf{Training for image recognition.} \textbf{a)}~The network of optical Kerr resonators, with a structure reminiscent of a convolutional neural network, including local connectivity and pooling (downsampling), but without enforcing translational invariance of the weights. Note that the physical connections shown here are bidirectional. \textbf{b)}~Model accuracy during training on the MNIST handwritten digits dataset, for different values of the nonlinearity strength $g/\bar{\kappa}$; left: training accuracy, right: test accuracy.  \textbf{c)}~ Confusion matrix on the test-set for the model with $g/\bar{\kappa}=0.2$ , reaching $97.4\%$ test accuracy.}
    \label{fig:imageRecognition}
\end{figure}
To further investigate the performance of our training scheme, we consider how a larger-scale  network of nonlinear Kerr-resonators can be trained with our method to perform image classification.
We employ an architecture that implements the local connectivity structure of  convolutional neural networks (CNNs), albeit for simplicity omitting both translational weight sharing and multiple channels (see Methods). In total, our network consists of about $10^3$ nodes and $7\cdot 10^3$ independent trainable parameters. The pixels of the image are encoded in the real-valued optical amplitudes of the input light fields, while the real parts of the output amplitudes in the final layer are taken to be the logarithms of the output probabilities (logits). We assess the influence of the Kerr nonlinearity on the training success (see Fig.~\ref{fig:imageRecognition}), finding that the test accuracy increases significantly with rising nonlinearity, improving from $92.6\%$ (linear system) to about $97.4\%$.

\section{Experimental requirements}

To apply our training approach as stated, it is not necessary to have full control or knowledge of the system.
We only require that the trainable parameters enter the linear parts of the equations of motion, e.g. as in Eq.~\eqref{eqn:nonl-system}, and that these trainable parts are known and accessible.
In optical systems, linear components can easily be tuned via heaters (time scales of 100 microseconds or more), phase-change materials (micro- to milliseconds) or electro-optically (100 picoseconds). At each of the tuneable elements, one needs optical readout to determine the scattering matrices for the training gradient extraction (e.g. via grating tap monitors~\cite{pai2023experimentally} attached to integrated resonators), and one must be able to inject light into the input and output nodes of the whole setup.

A wide variety of optical platforms are in principle suited, e.g., nonlinear integrated photonics devices such as those based on the promising recently developed thin-film LiNbO platform~\cite{zhu2021integrated}, Kerr-nonlinearities in coupled microresonators~\cite{ghalanos2020kerr}, exciton-polaritons in arrays of micropillars~\cite{fontaine2022kardar}, or systems with strong optical nonlinearities induced by atoms (e.g.~\cite{clark2020observation}). Input power levels of less than $100{\rm \mu W}$ can generate strong nonlinearities in on-chip microcavities~\cite{yoshiki2014all}.

Our method applies beyond systems strictly described by the equations of motion of Eq.~\eqref{eqn:nonl-system}, e.g. it extends to reciprocal continuous-wave systems, as can be seen by discretizing them.
Beyond that, our method also applies to the most general scattering setup where we consider right- and left-moving scattering waves propagating through a setup, e.g. of ring resonators and waveguides (see Methods and SI).

\section{Conclusion}
In summary, with Scattering Backpropagation we introduced a new training method which applies to a wide range of systems for which previously no efficient physics-based training approach existed. Specifically, our method applies to driven-dissipative nonlinear systems and is particularly relevant for training nonlinear optical systems with dissipation---one of the most promising neurmorphic platform. This alleviates the need for the often unsuccessful in-silico training, the inefficient parameter-shift method or hybrid approaches.
Gradients are computed by comparing only two scattering experiments. Remarkably, we neither require a faithful model, nor full control over the system, nor full knowledge of all of the system details; it is only necessary to access the scattering response at the positions of parameter updates.

Our method opens the door to the flexible experimental exploration of neuromorphic architectures with a wide range of systems which could previously not be considered due to the lack of efficient training methods.

%TC:ignore

\section{Methods}
\subsection{Linearized scattering matrix}
 Here, we derive the form of the linearized scattering matrix. For further details and general input-output relations, we refer to the SI. We consider equations of motion for $\bm\xi \coloneqq (\bm a, \bm a^*)^\mathsf{T}$ of the form $\dot{\bm\xi} = \bm F_\theta(\bm\xi)-\sqrt{\kappa}\, \bm\xi_\mathrm{in}$. For these equations, from here on, $\kappa=\diag(\kappa_1,\dots, \kappa_N, \kappa_1, \dots, \kappa_N)$ is the $2N\times 2N$ matrix defined by repeating the losses with respect to $\bm a$ and $\bm a^*$ on the diagonal. 
The steady state $\bar{\bm \xi}$ is the solution of
\begin{equation}\label{eq:steadyStateFree}
    \bm F_\theta(\bar{\bm\xi})=\sqrt{\kappa}\, \bm\xi_\mathrm{in}.
\end{equation}
If we perturb the input, the equation becomes $\dot{\bm \xi} = \bm F_\theta(\bm \xi)-\sqrt{\kappa}\, (\bm \xi_\mathrm{in}+\delta \bm \xi_\mathrm{in})$ and evolves towards a new steady state $\bar{\bm \xi}+\delta \bar{\bm \xi}$ which solves
\begin{equation}
    \bm F_\theta(\bar{\bm\xi}+\delta \bar{\bm \xi})=\sqrt{\kappa}\, (\bm\xi_\mathrm{in}+\delta \bm\xi_\mathrm{in}).
\end{equation}
By subtracting the two equations and expanding the vector field around $\bar{\xi}$ we obtain
\begin{equation}\label{eq:FreeMinusNudged}
   \sqrt{\kappa}\,  \delta \bm \xi_\mathrm{in} = \bm F_\theta(\bar{\bm\xi}+\delta \bar{\bm\xi}) - \bm F_\theta(\bar{\bm\xi}) = \nabla_{\bm\xi} \bm F_\theta(\bar{\bm\xi})\, \delta\bar{\bm\xi} + \mathcal{O}(\delta \bar{\bm\xi}^2),
\end{equation}
where $\nabla_{\bm\xi} \bm F_\theta$ indicates the Jacobian matrix.
Solving Eq.~\eqref{eq:FreeMinusNudged} for $\delta\bar{\bm\xi}$ and inserting the expression into the input-output relations $\bm\xi_\mathrm{out}=\bm\xi_{\mathrm{in}}+\sqrt{\kappa} \, \bm\xi$, we find for $\delta\bm\xi_\mathrm{out}$
\begin{equation}\label{eq:linearizedScatteringMatrix}
    \delta \bm \xi_\mathrm{out} = S_\theta(\bar{\bm \xi}) \,\delta \bm \xi_\mathrm{in} + \mathcal{O}(\delta \bm \xi_\mathrm{in}^2/\sqrt{\bar{\kappa}}),
\end{equation}
in which $\bar{\bm \xi}$ is the steady state and $ S_\theta(\bar{\bm \xi})\coloneqq\mathbf{I}_{2N} + \sqrt{\kappa}\, \nabla_{\bm \xi} \bm F_\theta(\bar{\bm \xi})^{-1}\sqrt{\kappa}$ is the linearized scattering matrix ---describing scattering according to the equations linearized around the nonlinear steady state.

\subsection{Gradient Approximation}
In the main text, we introduced the gradient approximation Eq.~\eqref{eq:general_gradient_formula} which, in our optical framework, leads to the gradient updates w.r.t. the frequencies, Eq.~\eqref{eq:gradientFreuquencies}, and the couplings, Eq.~\eqref{eq:gradientCouplings}.
In the following, we derive Eq.~\eqref{eq:general_gradient_formula}, i.e.
we show that for a system evolving according to
$\dot{\bm\xi} = \bm F_\theta(\bm\xi)-\sqrt{\kappa}\, \bm\xi_\mathrm{in}$ with input--output relations  $\bm\xi_\mathrm{out} =  \bm\xi_\mathrm{in} + \sqrt{\kappa} \,  \bm\xi$ (see SI for more general results), the expression for the gradient at the steady state $\bar{\bm\xi}$ in presence of a quasi-symmetry
\begin{equation}
\label{eqn:quasi_symm}
S_\theta(\bar{\bm \xi})^\dagger = US_\theta(\bar{\bm \xi})U^{-1}+\mathcal{O}(g/\bar{\kappa})
\end{equation}
can be expressed as
\begin{align}
    \frac{\partial C}{\partial \theta}
        &= - \left(\frac{\partial \bm F_\theta}{\partial  \theta }\right)^\dagger \sqrt{\kappa^{-1}} U \; \frac{\delta \bm \xi_\mathrm{out} - \delta \bm \xi_\mathrm{in}}{\beta} + \mathcal{O}\left(\frac{g}{\bar{\kappa}^2}, \frac{\beta}{\bar{\kappa}^2} \right),
\end{align}
where the derivatives are evaluated at $\bar{\bm\xi}$, $U$ is an invertible matrix, and
\begin{equation}
    \delta \bm \xi_\mathrm{in} \coloneqq \beta U^{-1}\frac{\partial C}{\partial \bm \xi^*_\mathrm{out}}(\mathbf{y},\mathbf{y}_\mathrm{target}).
\end{equation}
To show this, we differentiate the cost function $C$ w.r.t. a parameter $\theta$ applying the chain rule
\begin{align}\label{eq:derivativeCost}
    \frac{\partial C}{\partial \theta} &= \left(  \frac{\partial C}{\partial \bm \xi_\mathrm{out}} \right)^\mathsf{T} \frac{\partial \bm  \xi_\mathrm{out}}{\partial  \theta }.
\end{align}
Since the right-hand side is a scalar and both the cost function $C$ and the parameter $\theta$ are real, we can write
\begin{align}\label{eq:derivativeCost2}
    \frac{\partial C}{\partial \theta}
    & =
    \left(\frac{\partial \bm \xi_\mathrm{out}}{\partial  \theta }  \right)^\dagger \left( \frac{\partial C}{\partial \bm \xi_\mathrm{out}}\right)^*
\end{align}
Next, we derive an expression for $\frac{\partial \bm \xi_{\mathrm{out}}}{\partial  \theta }$.
Differentiating Eq.~\eqref{eq:steadyStateFree}, using the implicit function theorem, and applying the input-output relations $\bm \xi_{\mathrm{out}} =  \bm\xi_{\mathrm{in}} + \sqrt{\kappa} \,  \bm\xi$, we obtain
\begin{equation}\label{eq:xiOut}
    \frac{\partial \bm \xi_\mathrm{out}}{\partial  \theta} (\bar{\bm \xi}, \theta) = (\mathbf{I}_{2N} - S_\theta(\bar{\bm \xi}) )\sqrt{\kappa^{-1}} \frac{\partial \bm F_\theta}{\partial  \theta }(\bar{\bm \xi}),
\end{equation}
in which $ S_\theta(\bar{\bm \xi})\coloneqq\mathbf{I}_{2N} + \sqrt{\kappa} \,\nabla_{\bm\xi} \bm F_\theta(\bar{\bm \xi})^{-1}\sqrt{\kappa}$.
Combining Eqs.~\eqref{eq:derivativeCost2} and~\eqref{eq:xiOut}, we have
\begin{align}
    \frac{\partial C}{\partial \theta}
    & =
    \left(\frac{\partial \bm F_\theta}{\partial  \theta }\right)^\dagger \sqrt{\kappa^{-1}}(\mathbf{I}_{2N} - S_\theta(\Bar{\xi}))^\dagger\frac{\partial C}{\partial \bm \xi^*_\mathrm{out}}.
\end{align}
Next, using the quasi-symmetry Eq.~\eqref{eqn:quasi_symm} we find
\begin{align}
\label{eq:gradients_final_1}
    \frac{\partial C}{\partial \theta}
    & =
    \left(\frac{\partial \bm F_\theta}{\partial  \theta }\right)^\dagger \sqrt{\kappa^{-1}} U(\mathbf{I}_{2N} - S_\theta(\Bar{\bm \xi})) \underbrace{U^{-1}\frac{\partial C}{\partial \bm \xi^*_\mathrm{out}}}_{=: \delta \bm \xi_\mathrm{in}/\beta} + \mathcal{O}\left(\frac{g}{\bar{\kappa}^2}\right).
\end{align}
Note that, on the one hand, we want to arrange $\partial_\theta C$ in the form ``Scattering matrix" $\times$ ``input signal", which is a quantity we are able to extract via physical experiment. On the other hand, we want the use only a single input signal $\delta \bm \xi_\mathrm{in}$ to obtain the gradient w.r.t. all parameters $\theta$ simultaneously, as we aim to perform the experiment in the feedback phase just once (and not a number of times which depends on the number trainable parameters as in the ``parameter shift method"). As shown above, to solve the latter problem one can first take the adjoint in Eq.~\eqref{eq:derivativeCost2}, so it is possible to define a single error signal $\delta \bm \xi_\mathrm{in}$. Nevertheless, this comes at the cost of introducing $S_\theta(\bar{\bm{\xi}})^\dagger$ and so one has to make use of a quasi-symmetry Eq.~\eqref{eqn:quasi_symm} to re-obtain an expression of the form ``Scattering matrix" $\times$ ``input signal".

Finally, using Eq.~\eqref{eq:linearizedScatteringMatrix} in  Eq.~\eqref{eq:gradients_final_1}, we conclude 
\begin{align}
    \frac{\partial C}{\partial \theta}
    & = 
    - \left(\frac{\partial \bm F_\theta}{\partial  \theta }\right)^\dagger \sqrt{\kappa^{-1}}U \; \frac{\delta \bm \xi_\mathrm{out} - \delta \bm \xi_\mathrm{in}}{\beta} +  \mathcal{O}\left(\frac{g}{\bar{\kappa}^2}, \frac{\beta}{\bar{\kappa}^2} \right).
\end{align}

Let us now consider our system Eq.~\eqref{eqn:nonl-system}, and assume the quasi-symmetry Eq.~\eqref{eqn:quasi_symm} with $U=\sigma_y$. Let $\theta$ refer to detunings $\Delta_j$ and couplings $J_{j,\ell}$. In that case, the previous equation leads to approximations Eq.~\eqref{eq:gradientFreuquencies} and Eq.~\eqref{eq:gradientCouplings}. Similar results can be shown for other quasi-symmetry, i.e. different $U$ (SI). 

It is also possible to prove that the angle between the true gradient $\partial_\theta C$ and its approximation computed as above, depends on the angle $\alpha$ between $S_\theta(\bar{\bm \xi})^\dagger$ and $US_\theta(\bar{\bm \xi})U^{-1}$, which, in our system for either $U=\sigma_y$ or $U=\sigma_x$, is $\alpha=\mathcal{O}(g/\bar{\kappa})$ as shown in Fig.~\ref{fig:XOR} \textbf{b} and proven in the SI.

\subsection{Application to optical systems and quasi-reciprocity}
In order to study the steady state regime of the system Eq.~\eqref{eqn:nonl-system}, i.e. its linearization around $\bar{\bm a}$, we have to work in the $(\bm x,\bm p)$--quadrature basis or, equivalently, to consider the modes $\bm a$ and their conjugates $\bm a^*$ separately. This is because in the linearized regime there will be coupling between $\delta \bm a$ and $\delta \bm a^*$ as the nonlinear function $\bm \varphi$ is usually a non-holomorphic function of $\bm a$ (SI).
Thus, we linearize the system
\begin{equation}
\begin{cases}
    \label{eqn:system2'}
    \dot{\bm a} = -iH(\theta)\bm a  - i g \bm \varphi(\bm a, \bm a^*) - \sqrt{\kappa}\, \bm a_\mathrm{in}(\mathbf{x})\\
    \dot{\bm a}^* = iH^*(\theta)\bm a^* +i g [\bm \varphi(\bm a,\bm a^*)]^* - \sqrt{\kappa}\, \bm a^*_\mathrm{in}(\mathbf{x}).
\end{cases}
\end{equation} 
at steady state $(\bar{\bm a},\bar{\bm a}^*)$ obtaining
\begin{equation}
    \frac{d}{dt}
    \begin{pmatrix}
        \delta \bm a\\
        \delta \bm a^*
    \end{pmatrix} =
\nabla_{(\bar{\bm a},\bar{\bm a}^*)}\bm F_\theta(\bar{\bm a},\bar{\bm a}^*)
    \begin{pmatrix}
        \delta \bm a\\
        \delta \bm a^*
    \end{pmatrix},
\end{equation}
where the Jacobian $\nabla_{(\bm a,\bm a^*)}\bm F_\theta(\bar{\bm a},\bar{\bm a}^*)$ is
\begin{equation}
        \begin{pmatrix}
        -iH(\theta) -ig \frac{\partial \bm \varphi}{\partial \bm a }(\bar{\bm a},\bar{\bm a}^*) & -ig \frac{\partial \bm \varphi}{\partial \bm a^* }(\bar{\bm a},\bar{\bm a}^*)\\
        ig \frac{\partial \bm \varphi^*}{\partial \bm a }(\bar{\bm a},\bar{\bm a}^*) & iH^*(\theta) +ig \frac{\partial \bm \varphi^*}{\partial \bm a^* }(\bar{\bm a},\bar{\bm a}^*)
    \end{pmatrix}
\end{equation}
and the $\partial_{\bm{a}}$ symbol indicates the Wirtinger derivative with respect to $\bm a$ \cite{remmert1991theory}. Using the fact that 
\begin{equation}
    \frac{\partial \bm \varphi^*}{\partial \bm a} = \left( \frac{\partial \bm \varphi}{\partial \bm{a}^*} \right)^* \quad \text{and} \quad  \frac{\partial \bm \varphi^*}{\partial \bm a^*} = \left( \frac{\partial \bm \varphi}{\partial \bm a} \right)^*,
\end{equation}
we have that the Jacobian matrix has the form of a Bogoliubov transformation (without the usual normalization):
\begin{equation}
\label{eqn:Jac}
   \nabla_{(\bm a,\bm a^*)}\bm F_\theta(\bar{\bm a},\bar{\bm a}^*) = \begin{pmatrix}
    A(\bar{\bm a},\bar{\bm a}^*) & g B(\bar{\bm a},\bar{\bm a}^*)\\
    g B^*(\bar{\bm a},\bar{\bm a}^*) & A^*(\bar{\bm a},\bar{\bm a}^*)
\end{pmatrix},
\end{equation}
where $A(\bar{\bm a},\bar{\bm a}^*)$ and $B(\bar{\bm a},\bar{\bm a}^*)$ are $N \times N$ matrices depending on the steady state. As we discussed in general above, with Eq.~\eqref{eqn:Jac} one can introduce the linearized scattering matrix $S_\theta(\bar{\bm a},\bar{\bm a}^*)\coloneqq \mathbf{I}_{2N}+\sqrt{\kappa}\nabla_{(\bm a,\bm a^*)}\bm F_\theta(\bar{\bm a},\bar{\bm a}^*)^{-1}\sqrt{\kappa}$.

Note that, in the linear case $(g/\bar{\kappa}=0),$ the classical Hamiltonian we consider to write the dynamical equations Eq.~\eqref{eqn:nonl-system} via
\begin{equation}
    \dot{a}_j(t)=-\frac{\kappa_j}{2}{a}_j-i\frac{\partial \mathcal{H}}{\partial a^*_j}-\sqrt{\kappa_j} \,{a}_{\mathrm{in},j}
\end{equation}
is of the form
\begin{equation}
    {\mathcal{H}}(\bm{a},\bm{a}^*)= \sum_{j=1}^N \Delta_j a_j^* a_j +
    \sum_{j\not=\ell} J_{j,\ell} {a}_j^*{a}_\ell,
\end{equation}
and the linearized scattering matrix $S_\theta(\bar{\bm a},\bar{\bm a}^*)$ (which no longer depends on the steady state coordinates) has the following symmetries:
    \begin{equation}
    \label{eqn:s1'}
        S_\theta^\dagger = \sigma_y S_\theta \sigma_y,
    \end{equation}
    and 
    \begin{equation}
        \label{eqn:s2'}
        S^\dagger_\theta = \sigma_x S_\theta \sigma_x.
    \end{equation}
Recall that $\sigma_x$ flips $\bm a$ and $\bm a^*$.  From this, it follows that the system is reciprocal, i.e. light equally scatters in both directions between two different nodes $j$ and $\ell.$
In fact, for the Hamiltonian above, one can show that
\begin{equation}
    \begin{pmatrix}
        \bm a_{\mathrm{out}}\\
        \bm a_{\mathrm{out}}^*
    \end{pmatrix}=\underbrace{
    \begin{pmatrix}
       \Tilde{S}_\theta & 0\\
        0 & \Tilde{S}^*_\theta
    \end{pmatrix} 
    }_{S_\theta}    \begin{pmatrix}
        \bm a_{\mathrm{in}}\\
        \bm a_{\mathrm{in}}^*
    \end{pmatrix},
\end{equation}
where $\Tilde{S}_\theta= \mathbf{I}_N +i\sqrt{\kappa}H(\theta)^{-1}\sqrt{\kappa}$ is the $N\times N$ scattering matrix (at the zero frequency) as usually defined in linear systems. Note that $\Tilde{S}_\theta=\Tilde{S}_\theta^\mathsf{T}$ implies any of \eqref{eqn:s1'} and \eqref{eqn:s2'} and vice-versa.

If the system is nonlinear, such symmetry (reciprocity) is broken, nevertheless, for small values of $g,$ we observe only small deviations from reciprocity. In particular, corresponding to Eqs.~\eqref{eqn:s1'}, \eqref{eqn:s2'}, one can prove that the following quasi-symmetries hold (SI)
\begin{equation}
    \label{eqn:qs1'}
    S_\theta^{\dagger}(\bar{\bm a}, \bar{\bm a}^*) = \sigma_y S_\theta(\bar{\bm a}, \bar{\bm a}^*) \sigma_y + \mathcal{O}(g/\bar{\kappa}),
\end{equation}
\begin{equation}
    \label{eqn:qs2'}
    S_\theta^{\dagger}(\bar{\bm a},\bar{\bm a}^*) = \sigma_x S_\theta( \bar{\bm a},\bar{\bm a}^*) \sigma_x + \mathcal{O}(g/\bar{\kappa}).
\end{equation}
For instance, with respect to the Frobenius norm, one can prove that (SI):
\begin{align}
\notag
    &\|  S_\theta(\bar{\bm{a}},\bar{\bm{a}}^*)^\dagger - \sigma_y S_\theta(\bar{\bm{a}},\bar{\bm{a}}^*)\sigma_y \|_F \leq8\|\sqrt{\kappa}\|_F^2 \| H^{-1}\|_F  \cdot\\
    &\frac{g\|H^{-1}\|_F \left( \left\| \frac{\partial \varphi}{\partial \bm{a}}(\bar{\bm{a}},\bar{\bm{a}}^*)\right\|_F + \left\| \frac{\partial \varphi}{\partial \bm{a}^*}(\bar{\bm{a}},\bar{\bm{a}}^*)\right\|_F \right)}{1- 2 \sqrt{2}g\|H^{-1}\|_F \left( \left\| \frac{\partial \varphi}{\partial \bm{a}}(\bar{\bm{a}},\bar{\bm{a}}^*)\right\|_F + \left\| \frac{\partial \varphi}{\partial \bm{a}^*}(\bar{\bm{a}},\bar{\bm{a}}^*)\right\|_F \right)}.
    \label{eq:big-O-terms}
\end{align}
Since relations Eq.~\eqref{eqn:s1'} and Eq.~\eqref{eqn:s2'} are equivalent to linear reciprocity $\Tilde{S}_\theta=\Tilde{S}_\theta^\mathsf{T}$, we informally say that a nonlinear system Eq.~\eqref{eqn:nonl-system} is \emph{quasi-reciprocal} if it is in a steady state regime where such relations well approximate Eq.~\eqref{eqn:qs1'} and Eq.~\eqref{eqn:qs2'}, respectively.

\subsection{Generalization to arbitrary linear input-output relations}

For clarity of presentation, we have so far considered systems of the form Eq.~\eqref{eq:dyn_sys}, with input–output relations $\bm{\xi}_\mathrm{out} = \bm{\xi}_\mathrm{in} + \sqrt{\kappa}\, \bm{\xi}$ and a quasi-symmetry Eq.~\eqref{eqn:quasi-symmetry-1}. In the SI, we formulate our method in a more general setting. Specifically, we consider a dynamical system $\dot{\bm{\xi}}=\bm{F}_\theta(\bm{\xi})-\Pi \, \bm{\xi}_\mathrm{in},$ with linear input–output relations $\bm{\xi}_\mathrm{out} = \Gamma \bm{\xi}_\mathrm{in} + \Sigma\, \bm{\xi}$ and quasi-symmetry $(\nabla_{\bm\xi} \bm{F}_\theta (\bar{\bm\xi})^{-1})^\dagger = U_1 \nabla_{\bm\xi} \bm{F}_\theta(\bar{\bm{\xi}})^{-1}U_2+\mathcal{O}(g)$, in which $\Pi, \Gamma, \Sigma, U_1$, and $U_2$ are square matrices. In contrast to $\sqrt{\kappa}$, $\Pi$ and $\Sigma$ are not necessarily diagonal matrices.
This generalization is crucial for addressing, for instance, optical systems consisting of components coupled by waveguides which support waves propagating in both directions. Elimination of the waveguides (via the standard input-output equations) leads to input-output relations for the external signals of the generalized type introduced here.  In the SI, we specifically consider the example of transmission in optical ring resonators coupled sequentially via a waveguide and successfully apply Scattering Backpropagation.

\subsection{Numerical simulations}

\emph{\textbf{Training XOR}}
To showcase supervised training with Scattering Backpropagation, in the main text, we consider a neuromorphic network of three coupled Kerr non-resonators with $g/\bar{\kappa}=0.2$, represented in Fig.~\ref{fig:XOR}~\textbf{a}. We asses the regression task of learning the XOR binary function, $\oplus:\{0,1\}^2\to\{0,1\}$ such that $0\oplus0=1\oplus1=0$ and $1\oplus0=0\oplus1=1.$
In this case, the network input $\mathbf{x}=(\mathrm{x}_1,\mathrm{x}_2)^\mathsf{T} \in \mathcal{D}_\mathrm{train}\coloneqq\{0,1\}^2$ is encoded in the real parts of the input signal, i.e. $\mathfrak{Re}(a_{\mathrm{in},1})/\bar{\kappa}\coloneqq \mathrm{x}_1$ and $\mathfrak{Re}(a_{\mathrm{in},2})/\bar{\kappa}\coloneqq \mathrm{x}_2$, while their imaginary parts are set to zero. The output is read from the real part of $a_{\mathrm{out},3}$, in particular $\mathbf{y} \coloneqq \, 10  \cdot \mathfrak{Re}(a_{\mathrm{out},3})/\bar{\kappa}$. Thus, the indices of the input and output nodes are respectively $\mathcal{I}_\mathrm{in} \coloneqq \{1,2\}$ and $\mathcal{I}_\mathrm{out} \coloneqq \{3\}$.
We initialize the trainable weights following Xavier's convention \cite{glorot2010understanding} while, as cost function, we use the mean-squared error $C(\mathbf{y},\mathbf{y}_\mathrm{target})=\frac{1}{4}\sum_{\mathbf{x}\in\mathcal{D}_\mathrm{train}}(\mathbf{y}(\mathbf{x})- \mathbf{y}_\mathrm{target})^2$. During training, we numerically solve the dynamical equations up to $\bar{\kappa} \, t_\mathrm{max}=30$ using stepsize $\bar{\kappa} \, dt=0.01.$ Furthermore, we choose $\beta/\bar{\kappa} = 0.01$ and learning rate $\eta=10^{-3}$ for the weight update via Eq.~\eqref{eq:gradientFreuquencies} and Eq.~\eqref{eq:gradientCouplings}.
In Fig.~\ref{fig:XOR}~\textbf{b}, we plot the cost function evolving over 200 epochs, each consisting of training over the entire dataset $\mathcal{D}_{\mathrm{train}}$, and in Fig.~\ref{fig:XOR}~\textbf{c} we show the time evolution of the $\bm{a}_\mathrm{out}(t)$ modes in a trained model.
In the SI, we compare the same three modes self-Kerr architecture on XOR for different values of $g/\bar{\kappa}$. Furthermore, we also consider a larger network of $N=10$ modes with cross-Kerr nonlinearities. As a side remark,  our results show how a system described by Eq.~\eqref{eqn:nonl-system} with self-Kerr nonlinearities is able to learn XOR with just $N=3$ modes, and so only $6$ real independent parameters. For comparison, a Hopfield network needs at least $4$ nodes, and $10$ real independent parameters  \cite{rojas2013neural}, while a multi-layer perceptron with Tanh activation requires at least $5$ nodes (two hidden) and $8$ parameters.

\emph{\textbf{Approximation Analysis.}}
The angle between the true and the approximated gradient given by Eqs.~\eqref{eq:gradientFreuquencies} and \eqref{eq:gradientCouplings} depends the one between $S_\theta(\bar{\bm{a}},\bar{\bm{a}}^*)^\dagger$ and $\sigma_y S_\theta(\bar{\bm{a}},\bar{\bm{a}}^*) \sigma_y$ (SI). For a fully-connected network of self-Kerr resonators, we investigated how such a quantity (defined with respect to the Frobenius inner product) varies with respect to $g |a^{0}_\mathrm{in}|^2/\bar{\kappa}^2$, where $|a^{0}_\mathrm{in}|$ is the average input strength in a site, see Fig.~\ref{fig:XOR}~\textbf{f} and Fig.~\ref{fig:XOR}~\textbf{g}. Indeed, note that the order of $\left\| \frac{\partial \varphi}{\partial \bm{a}}(\bar{\bm{a}},\bar{\bm{a}}^*)\right\|_F$ and $\left\| \frac{\partial \varphi}{\partial \bm{a}^*}(\bar{\bm{a}},\bar{\bm{a}}^*)\right\|_F $ in Eq.~\eqref{eq:big-O-terms} is $|a^{0}_\mathrm{in}|^2$ for self-Kerr nonlinearities (SI). Numerically, to compute the linearized scattering matrix $S_\theta(\bar{\bm{a}},\bar{\bm{a}}^*)$ for different system sizes $N$, input strength and initial parameters $\Delta$ and $J$, we solved the equilibrium equation with the \texttt{fsolve()} function of the \texttt{scipy.optimize.module} \cite{virtanen2020scipy}. The same choice was made when training the $N=10$ network on XOR to investigate how $\beta$ affects the final training accuracy when fixing the initial trainable parameters $J$, i.e. Fig.~\ref{fig:XOR}~\textbf{e}. Furthermore, in both cases, in the simulations we took $\kappa_j=\bar{\kappa}$ for every $j$ and set the internal losses $\kappa_j'$ to zero.

\emph{\textbf{Training MNIST.}} To investigate the performance of our method on a more complex benchmark, we train a network of self-Kerr modes to perform image classification on the MNIST dataset, consisting of $28 \times 28$ pixel images of hand-written digits from $0$ to $9.$ The dataset consists of 60,000 images in the training set, and 10,000 images in the test set. Inspired by CNNs, we set up a layered architecture having sparse connections, similar to what was done in \cite{Wanjura2024Fully}. To keep the network structure simple (and more experimentally plausible), we choose not to introduce multiple channels. To compensate for the resulting reduction in degrees of freedom, we do not implement translational weight sharing between kernels acting on different locations. This actually mimics the brain's visual cortex structure more closely than a standard CNN.  The physical neurons are subdivided into four (one input, two hidden, and one output) layers of shape $(28 \times 28)-(12 \times 12)-(5 \times 5)-10$. The connectivity in the first two connection layers is sparse according to square kernels of respectively size $6$ and $4$ (both with stride $2)$, while the final connection layer is dense. In this way, the network is able to capture the local patterns in the image with a limited number of trainable parameters (compared to a fully connected architecture).  In total, our network consists of $N=963$ nodes and $6,797$ independent trainable parameters, namely the resonators' detunings and couplings. The pixels of the image are collected in an input vector $\mathbf{x}$ whose components are encoded in the real parts of the input light fields incident on the nodes in the first network layer. Specifically, we set $\mathfrak{Re}(a_{\mathrm{in}, j})/\bar{\kappa} \coloneqq \frac{\mathrm{x}_j}{100\sqrt{2}}$ for $j$ in $\mathcal{I}_\mathrm{in} \coloneqq \{1, \dots, 784\}$, where the factor of $1/100$ is chosen to rescale the input pixels to be in the order of $1$ and the $1/\sqrt{2}$ because the code implementation is in term of the field (real) quadratures. Then, we consider the real part of the output light in the final layer $a_{\mathrm{out},j}$ for $j$ in $\mathcal{I}_\mathrm{out} \coloneqq \{954, \dots, 963\}$ to be the ten logits which are the output $\mathbf{y}$ of the network.
Furthermore, consider a cross-entropy loss function $C(\mathbf{y},\mathbf{y}_\mathrm{target})=-\sum_{m=1}^{10}\mathbf{y}_{\mathrm{target},m}\log(\sigma(\mathbf{y})_m)$, in which $\sigma(\mathbf{y})_m=\frac{\exp(\mathrm{y}_m/T)}{\sum_{k=1}^{10} \exp( \mathrm{y}_k/T)}$ is the softmax function with temperature $T=0.1$, and $\mathbf{y}_\mathrm{target}$ is the one-hot encoding of $\mathbf{x}$'s true label.
During training, we numerically solve the dynamical equations \eqref{eqn:nonl-system} up to $\bar{\kappa} \, t_\mathrm{max}=60$ using a stepsize $\bar{\kappa} \, dt=0.1.$ We choose $\beta/\bar{\kappa} = 0.01$ in Eqs.~\eqref{eq:gradientFreuquencies},~\eqref{eq:gradientCouplings}, and a learning rate of $\eta=0.1$ for the weight update Eq.~\eqref{eqn:weight-update}. We perform stochastic gradient descent, averaging the approximated gradients over mini-batches of size $10$.

\subsection{Measuring the exact gradient with $2N$ scattering experiments}

As mentioned in the main text, even if the system is not quasi-reciprocal, it is possible to change the training procedure described in Section \ref{sec:extracting_gradients} to compute the exact gradient $\partial_\theta C$ without assuming any quasi-symmetry. This requires performing $2N$ scattering experiments in the feedback phase (instead of one, as we do in Scattering Backpropagation). In a large, fully connected network with trainable couplings $J_{j,\ell}$, this procedure is still much more efficient than the parameter-shift method, which involves $N^2$ experiments. The goal is to measure the linearized scattering matrix and compute $\partial_\theta C$ via Eq.~\eqref{eq:xiOut} ---that we obtained differentiating the steady state Eq.~\eqref{eq:steadyStateFree} using the implicit function theorem.

First, recall that, in the case of our dynamical equations Eq.~\eqref{eqn:system2'} with $\bm{\xi}\coloneqq(\bm{a}, \bm{a}^*)^\mathsf{T}$, the linearized scattering matrix of the system  has the form of a Bogoliuvov transformation
\begin{equation}
    S_\theta(\bar{\bm a},\bar{\bm a}^*) = \begin{pmatrix}
        S_{11} & S_{12}\\
        S_{12}^* & S_{11}^*
    \end{pmatrix},
\end{equation}
where $S_{11}$ and $S_{12}$ are $N\times N$ matrices. Now, the main idea is to repeat the feedback phase described in Section \ref{sec:extracting_gradients} $2N$ times. In particular, for each $k=1,\dots, N$, we perform two scattering experiments: first, we define the error signal to be $\delta \bm{a}^{(2k-1)}_{\mathrm{in}}\coloneqq(0, \dots, 0, \beta', 0, \dots, 0)^\mathsf{T}$, in which the non-zero entry is the $k$-th and $\beta'$ is a small, positive number (in units of the square root of a loss rate). In this way, the system response determined by Eq.~\eqref{eq:in-out-relations} at the $(2k-1)$-th iteration is given by
\begin{equation}
    \delta \bm{a}^{(2k-1)}_{\mathrm{out}}=\beta' \,\Bigl( S_{11}^{(k)} + S_{12}^{(k)}\Bigr)+\mathcal{O}(\beta'^2/\sqrt{\bar{\kappa}}),
\end{equation}
where we indicated as $\delta \bm{a}^{(2k-1)}_{\mathrm{out}}$ the perturbation on the output after the $(2k-1)$-th experiment, and with $S_{11}^{(k)}$ the $k$-th column of the matrix $S_{11}$.
Next, in the $(2k)$-th iteration, we define a new error signal to be $\delta \bm{a}^{(2k)}_{\mathrm{in}}\coloneqq(0, \dots, 0, i\beta', 0, \dots, 0)^\mathsf{T}$, in which the non-zero entry is the $k$-th, and measure the system response
\begin{equation}
     \delta \bm{a}^{(2k)}_{\mathrm{out}}=i\beta' \,\Bigl( S_{11}^{(k)} - S_{12}^{(k)}\Bigr)+\mathcal{O}(\beta'^2/\sqrt{\bar{\kappa}}).   
\end{equation}
Therefore, a column at the time, it is possible to recover the full linearized scattering matrix (up to $\mathcal{O}(\beta')$ terms) using
\begin{equation}
    S_{11}^{(k)} = \frac{\delta \bm{a}^{(2k-1)}_{\mathrm{out}} - i\,\delta \bm{a}^{(2k)}_{\mathrm{out}}}{2 \beta'}+\mathcal{O}(\beta'/\sqrt{\bar{\kappa}})
\end{equation}
and
\begin{equation}
    S_{12}^{(k)} = \frac{\delta \bm{a}^{(2k-1)}_{\mathrm{out}} + i\,\delta \bm{a}^{(2k)}_{\mathrm{out}}}{2 \beta'}+\mathcal{O}(\beta'/\sqrt{\bar{\kappa}}),
\end{equation}
and finally compute the gradient $\partial_\theta C$ via Eq.~\eqref{eq:xiOut}.
Note that with this modified version of the algorithm, at the price of $2N$ experiments, we have reconstructed the full linearized scattering matrix $S_\theta(\bar{\bm a},\bar{\bm a}^*)$ without any assumption on the quasi-reciprocity or on the system.
Furthermore, notice that if the trainable parameters $\theta$ of the system Eq.~\eqref{eqn:system2'} are the detunings $\Delta_j$ and the couplings $J_{j,\ell}$, this modified method performs $\mathcal{O}(\sqrt{N_{\theta}})$ experiments, which is much more efficient than the parameter-shift method requiring $\mathcal{O}(N_{\theta})$ for a fully connected setup, where $N_{\theta}$ is the number of parameters. 

Nevertheless, in the parameter regime in which the system of Kerr-resonators we investigated numerically possesses a steady state (the one we are interested in), the gradient approximation given by Eqs.~\eqref{eq:gradientFreuquencies} and~\eqref{eq:gradientCouplings} was already sufficient for performing gradient descent on the considered tasks. In addition, the quasi-symmetries of Eq.~\eqref{eqn:system2'} depend on $g,$ the input power $|a^0_{\mathrm{in}}|$, and the losses $\kappa$ (see Fig.~\ref{fig:XOR} and SI). Thus, in many neuromorphic applications, it is probably more efficient to engineer such quantities to design a quasi-reciprocal system in the first place, and use the approximate version of Scattering Backpropagation rather than this less efficient alternative.

%\bibliographystyle{apsrev4-2}
%\bibliography{bib.bib}

\twocolumngrid
\onecolumngrid

\clearpage

\appendix

\setcounter{figure}{0}
\renewcommand{\thefigure}{S\arabic{figure}}
\renewcommand{\theequation}{S\arabic{equation}}
\setcounter{equation}{0}
\setcounter{section}{0}
\setcounter{table}{0}

\onecolumngrid
\part*{\large\centering Supplementary Information for Training nonlinear optical neural networks with Scattering Backpropagation}
\begin{center}
    Nicola Dal Cin, Florian Marquardt, Clara C. Wanjura
\end{center}

\section{Unit-less equations}
In the main text, we consider the dynamics modeling the time-evolution of complex modes, e.g. optical resonators, $a(t)=(a_1(t), \dots, a_N(t))^\mathsf{T}$ via
\begin{equation}
\label{eqn:system1}
    \frac{d}{dt}a(t) = -i H(\theta) a(t) - i g \varphi(a(t)) - \sqrt{\kappa}\, a_\mathrm{in}(\mathrm{x}),
\end{equation}
where $H_{j,\ell}\coloneqq J_{j,l}$ and $H_{j,j}\coloneqq-i\frac{\kappa_j+\kappa_j'}{2}+\Delta_j$. In particular, $\Delta_j \coloneqq J_{j,j}$ represents the detunings, $J$ is the real symmetric coupling matrix of the nodes, $\kappa'_j$ and $\kappa_j$ are respectively the internal and external (e.g. due to waveguide coupling) losses of node $a_j.$ As we present in the main text following the usual convention, the dynamical equations \eqref{eqn:system1} are in frequency units, thus the modes $a_j(t)$ are unit-less, while $1/t$, $\kappa_j$, $\kappa_j'$, $\Delta_j$, $J_{j,\ell}$ and $g$ are frequencies ($a_{\mathrm{in},j}$ and $a_{\mathrm{out},j}$ are in units of $\sqrt{\kappa_j}$ for convention). Moreover, as in the main text we define
\begin{equation}
    \delta a_\mathrm{in}\coloneqq -i\beta \frac{\partial C}{\partial a_\mathrm{out}},
\end{equation}
in which we have that $\beta$ is also a frequency.
However, in order to work with unit-less equations, in this Supplementary Information we will rescale \eqref{eqn:system1} by a suitable reference rate $\bar{\kappa}$, introducing $\tilde{t}\coloneqq\bar{\kappa}t$
\begin{align}
    \frac{d a}{d \tilde{t}}(\tilde{t}/\bar{\kappa}) &= \frac{d a}{d \tilde{t}}(t) = \frac{da}{dt}(t)\frac{dt}{d\tilde{t}}(\tilde{t})=\frac{1}{\bar{\kappa}}\left( -i H(\theta) a(t) - i g \varphi(a(t)) - \sqrt{\kappa}\, a_\mathrm{in}(\mathrm{x})\right)\\
    &=-i\tilde{H}(\tilde{\theta})a(\tilde{t}/\bar{\kappa})-i\tilde{g}\varphi(a(\tilde{t}/\bar{\kappa}))-\sqrt{\tilde{\kappa}}\, \tilde{a}_\mathrm{in}(\mathrm{x}),
    \label{eq:part_res}
\end{align}
where
\begin{equation}
\label{eq:dimension-less}
    \tilde{H}_{j,\ell} \coloneqq \tilde{J}_{j, \ell} \coloneqq \frac{J_{j,\ell}}{\bar{\kappa}}, \quad \tilde{g}=\frac{g}{\bar{\kappa}}, \quad \tilde{\kappa}\coloneqq \frac{\kappa}{\bar{\kappa}}, \quad \tilde{\kappa}'\coloneqq \frac{\kappa'}{\bar{\kappa}}, \quad  \tilde{\Delta}_j \coloneqq \frac{\Delta_j}{\bar{\kappa}}, \quad \tilde{H}_{j,j}\coloneqq-i\frac{\tilde{\kappa}_j+\tilde{\kappa}_j'}{2}+\tilde{\Delta}_j, \quad\text{and} \quad \tilde{a}_\mathrm{in}(\mathrm{x})\coloneqq \frac{a_\mathrm{in}(\mathrm{x})}{\sqrt{\bar{\kappa}}}.
\end{equation}
By letting $\tilde{a}(\tilde{t})\coloneqq a(\tilde{t}/\bar{\kappa})$, equation \eqref{eq:part_res} can be expressed as
\begin{equation}
    \dot{\tilde{a}}(\tilde{t})\coloneqq \frac{d\tilde{a}}{d\tilde{t}}(\tilde{t})=-i\tilde{H}(\tilde{\theta})\tilde{a}(\tilde{t})-i\tilde{g}\varphi(\tilde{a}(\tilde{t}))-\sqrt{\tilde{\kappa}}\, \tilde{a}_\mathrm{in}(\mathrm{x}),
\end{equation}
which is unit-less and of the same form of \eqref{eqn:system1}. In this way, one also has
\begin{equation}
    \delta \tilde{a}_\mathrm{in}=\frac{\delta a_\mathrm{in}}{\sqrt{\bar{\kappa}}}= -i\frac{\beta}{\sqrt{\bar{\kappa}}} \frac{\partial C}{\partial a_\mathrm{out}}= -i\frac{\beta}{\bar{\kappa}}\frac{\partial C}{\partial \tilde{a}_\mathrm{out}},
\end{equation}
and so $\tilde{\beta}=\beta / \bar{\kappa}.$
From now on, with some abuse of notation, we will refer to this unit-less formulation omitting the `tildes' for readability.

\section{Linearized Scattering Matrix}
\label{app:lin_scatt_mat}
In order to study the steady-state regime of the system \eqref{eqn:system1}, i.e. its linearization around $\bar{a}$, we have to work in the $(x,p)$--quadrature basis or, equivalently, to consider the modes $a$ and their conjugates $a^*$ separately. This is because in the linearized regime there will be coupling between $\delta a$ and $\delta a^*$ as the nonlinear function $\varphi$ is usually a non-holomorphic function of $a$ (see Remark \ref{rem:nonholomorphic}). Thus, we will study a system of the form
\begin{equation}
\begin{cases}
    \label{eqn:system2}
    \dot{a} = -iH(\theta)a - i g \varphi(a, a^*) - \sqrt{\kappa}\, a_\mathrm{in}(\mathrm{x})\\
    \dot{a}^* = iH^*(\theta)a^* +i g [\varphi(a,a^*)]^* - \sqrt{\kappa}\, a^*_\mathrm{in}(\mathrm{x}),
\end{cases}
\end{equation}
with input-output relations
\begin{equation}
\label{eq:inout1}
    a_{\mathrm{out},j} = a_{\mathrm{in},j} + \sqrt{\kappa_j} a_j, \quad \quad a^*_{\mathrm{out},j} = a^*_{\mathrm{in},j} + \sqrt{\kappa_j} a^*_j, \quad \quad \text{for each   } j=1,\dots,N.
\end{equation}
In fact, it is convenient consider a more general system of differential equations
\begin{equation}
    \dot{\xi}=F(\xi)-\Pi \, \xi_\mathrm{in},
\end{equation}
with linear input-output relations
\begin{equation}
    \xi_\mathrm{out}=\Gamma \xi_\mathrm{in}+ \Sigma \xi,
\end{equation}
where $\xi(t) \in \R^{m}$, while $\Gamma,$ $\Pi,$ $\Sigma$ are $m \times m$ invertible matrices --- also assuming $\Pi,$ $\Sigma$ are invertible. 

Thus, for our case of $N$ optical modes described by \eqref{eqn:system2} and \eqref{eq:inout1}, we have $\xi (t)\coloneqq (a(t), a^*(t))^{\mathsf{T}}$, $\Gamma \coloneqq \mathbf{I}_{2N}$ and $\Pi = \Sigma \coloneqq \sqrt{\kappa}$ where,
with some abuse of notation, $\kappa=\diag(\kappa_1, \dots, \kappa_N,\kappa_1, \dots, \kappa_N)$ indicates the $2N\times 2N$ matrix defined by repeating the losses with respect to $a$ and $a^*$ on the diagonal. 

\begin{lem}[Linearized Scattering Matrix]
\label{lem:smatrix}
    Consider a system 
        $\dot{\xi} = F(\xi)-\Pi\, \xi_\mathrm{in}$ with linear input output relations $\xi_\mathrm{out}=\Gamma \xi_\mathrm{in}+ \Sigma \, \xi$. If, in the steady state regime, we perturb the input field by $\delta \xi_\mathrm{in}$, then we have
    \begin{equation}
        \delta \xi_\mathrm{out} = S(\bar{\xi}) \,\delta \xi_\mathrm{in} + \mathcal{O}(\delta \xi_\mathrm{in}^2),
    \end{equation}
    where $\bar{\xi}$ is the steady state and $ S(\bar{\xi})\coloneqq \Gamma + \Sigma \, \nabla_\xi F(\bar{\xi})^{-1}\Pi$ is the linearized scattering matrix.
\end{lem}
\begin{proof}
    The steady state of the free system $\bar{\xi}$ is the solution of
    \begin{equation}
        F(\bar{\xi})=\Pi \, \xi_\mathrm{in}.
    \end{equation}
If we perturb the input, the system becomes $\dot{\xi} = F(\xi)-\Pi\, (\xi_\mathrm{in}+\delta \xi_\mathrm{in})$ and evolves towards a new steady state $\bar{\xi}+\delta \bar{\xi}$ which solves
    \begin{equation}
        F(\bar{\xi}+\delta \bar{\xi})=\Pi\, (\xi_\mathrm{in}+\delta \xi_\mathrm{in}).
    \end{equation}
By subtracting the two equations and expanding the vector field around $\bar{\xi}$ we get
\begin{equation}
   \Pi\,  \delta \xi_\mathrm{in} = F(\bar{\xi}+\delta \bar{\xi}) - F(\bar{\xi}) = \nabla_\xi F(\bar{\xi})\, \delta\bar{\xi} + \mathcal{O}(\delta \bar{\xi}^2).
\end{equation}
Finally, we conclude by recalling the input--output relation $\delta \xi_\mathrm{out} = \Gamma \delta \xi_\mathrm{in} + \Sigma \, \delta \bar{\xi}$ and inverting the equation above.
\end{proof}

\section{Quasi-Reciprocity in optical systems}
\label{app:quasi-reciprocity}
Note that in our optical case described by \eqref{eqn:system2} and \eqref{eq:inout1}, the linearized scattering matrix takes the form $S(\bar{a},\bar{a}^*)=\mathbf{I}_{2N} + \sqrt{\kappa} \nabla_{(\bar{a},\bar{a}^*)}F(\bar{a},\bar{a}^*)\sqrt{\kappa}$. More explicitly, linearizing equations \eqref{eqn:system2} at steady state $(\bar{a},\bar{a}^*)$ leads to 
\begin{equation}
    \label{eqn:linearized_system2}
    \frac{d}{dt}
    \begin{pmatrix}
        \delta a\\
        \delta a^*
    \end{pmatrix} =
\nabla_{(\bar{a},\bar{a}^*)}F(\bar{a},\bar{a}^*)
    \begin{pmatrix}
        \delta a\\
        \delta a^*
    \end{pmatrix},
\end{equation}
where the Jacobian is
\begin{equation}
\label{eqn:jacobian}
        M \coloneqq \nabla_{(a,a^*)}F(\bar{a},\bar{a}^*)=\begin{pmatrix}
        -iH(\theta) -ig \frac{\partial \varphi}{\partial a }(\bar{a},\bar{a}^*) & -ig \frac{\partial \varphi}{\partial a^* }(\bar{a},\bar{a}^*)\\
        ig \frac{\partial \varphi^*}{\partial a }(\bar{a},\bar{a}^*) & iH^*(\theta) +ig \frac{\partial \varphi^*}{\partial a^* }(\bar{a},\bar{a}^*)
    \end{pmatrix}
\end{equation}
and the $\partial_a$ symbol indicates the Wirtinger derivative with respect to $a$ \cite{remmert1991theory}. 
Since for differentiable functions we have
\begin{equation}
    \frac{\partial f^*}{\partial a} = \left( \frac{\partial f}{\partial a^*} \right)^* \quad \text{and} \quad  \frac{\partial f^*}{\partial a^*} = \left( \frac{\partial f}{\partial a} \right)^*,
\end{equation}
it follows that the Jacobian matrix has the form of a Bogoliubov transformation (although without the usual normalization):
\begin{equation}
   M = \begin{pmatrix}
    A(\bar{a},\bar{a}^*) & g B(\bar{a},\bar{a}^*)\\
    g B^*(\bar{a},\bar{a}^*) & A^*(\bar{a},\bar{a}^*)
\end{pmatrix},
\end{equation}
where $A(\bar{a},\bar{a}^*)$ and $B(\bar{a},\bar{a}^*)$ are $N \times N$ matrices depending on the steady state. 
More generally, if $\mathcal{H}$ is the (real) classical Hamiltonian used to derive the dynamical equations 
\begin{equation}
    \dot{a}_j = -\frac{\kappa_j}{2}a_j - i \frac{\partial \mathcal{H}}{\partial a^*_j}-\sqrt{\kappa_j} a_{\mathrm{in},j}
\end{equation}
the Jacobian matrix of the latter (also considering the $a^*$ modes) can be also written as
\begin{equation}
    -\frac{1}{2}\begin{pmatrix}
        \kappa & 0\\
        0 & \kappa
    \end{pmatrix}-i \sigma_z \begin{pmatrix}
        \dfrac{\partial }{\partial a}\dfrac{\partial \mathcal{H} }{\partial a^*}(\bar{a},\bar{a}^*) & \dfrac{\partial }{\partial a^*}\dfrac{\partial \mathcal{H} }{\partial a^*}(\bar{a},\bar{a}^*)\\[8pt]
        \dfrac{\partial }{\partial a}\dfrac{\partial \mathcal{H} }{\partial a}(\bar{a},\bar{a}^*) & \dfrac{\partial }{\partial a^*}\dfrac{\partial \mathcal{H} }{\partial a}(\bar{a},\bar{a}^*)
    \end{pmatrix}, \quad \text{where} \quad \sigma_z \coloneqq \begin{pmatrix}
        \textbf{I}_N & 0\\
        0 & -\textbf{I}_N
    \end{pmatrix}.
\end{equation}
Furthermore, note that 
\begin{equation}
    \dfrac{\partial }{\partial a}\dfrac{\partial \mathcal{H} }{\partial a^*}(\bar{a},\bar{a}^*) = \left(\dfrac{\partial }{\partial a^*}\dfrac{\partial \mathcal{H} }{\partial a}(\bar{a},\bar{a}^*)\right)^*, \quad \quad \dfrac{\partial }{\partial a^*}\dfrac{\partial \mathcal{H} }{\partial a^*}= \left(\dfrac{\partial }{\partial a}\dfrac{\partial \mathcal{H} }{\partial a}(\bar{a},\bar{a}^*) \right)^*
\end{equation}
and that they are respectively Hermitian and symmetric matrices since
\begin{equation}
    \frac{\partial}{\partial a_\ell}\frac{\partial \mathcal{H}}{\partial a_j^*}(\bar{a},\bar{a}^*)=\frac{\partial}{\partial a_j^*}\frac{\partial \mathcal{H}}{\partial a_\ell}(\bar{a},\bar{a}^*)=\frac{\partial}{\partial a_j^*}\left(\frac{\partial \mathcal{H}^*}{\partial a_\ell}\right)(\bar{a},\bar{a}^*)=\frac{\partial}{\partial a_j^*}\left(\frac{\partial \mathcal{H}}{\partial a^*_\ell}\right)^*(\bar{a},\bar{a}^*)=\left( \frac{\partial}{\partial a_j} \frac{\partial \mathcal{H}}{\partial a^*_\ell}(\bar{a},\bar{a}^*) \right)^*
\end{equation}
and
\begin{equation}
    \frac{\partial}{\partial a^*_\ell}\frac{\partial \mathcal{H}}{ \partial a^*_j}(\bar{a},\bar{a}^*)= \frac{\partial }{\partial a^*_j }\frac{\partial \mathcal{H}}{ \partial a^*_\ell}(\bar{a},\bar{a}^*).
\end{equation}
In particular, in our case \eqref{eqn:system1} in which $\mathcal{H}$ has real, symmetric couplings $J_{j,\ell}$ and
\begin{align}
\label{eqn:Poisson_b}
    \dot{a}_j &= -\frac{\kappa_j}{2}a_j - i \frac{\partial \mathcal{H}}{\partial a^*_j}-\sqrt{\kappa_j} a_{\mathrm{in},j}\\
    &=-\frac{\kappa_j}{2}a_j - i\sum_{\ell=1}^NJ_{j,\ell}a_\ell-ig\varphi_j(a,a^*)-\sqrt{\kappa_j} a_{\mathrm{in},j},
\end{align} this implies that the matrices $\frac{\partial \varphi}{\partial a}(\bar{a},\bar{a}^*)$ and $\frac{\partial \varphi}{\partial a^*}(\bar{a},\bar{a}^*)$ are also respectively Hermitian and symmetric as
\begin{align}
    J_{j,\ell}+g\frac{\partial \varphi_j}{\partial a_\ell}(\bar{a},\bar{a}^*)=\frac{\partial}{\partial a_\ell}\frac{\partial \mathcal{H}}{\partial a_j^*}(\bar{a},\bar{a}^*)=\left( \frac{\partial}{\partial a_j} \frac{\partial \mathcal{H}}{\partial a^*_\ell}(\bar{a},\bar{a}^*) \right)^*=J_{\ell,j}^*+\left(g\frac{\partial \varphi_\ell}{\partial a_j}(\bar{a},\bar{a}^*)\right)^*
\end{align}
and
\begin{equation}
    g\frac{\partial \varphi_j}{\partial a^*_\ell}(\bar{a},\bar{a}^*)=\frac{\partial}{\partial a^*_\ell}\frac{\partial \mathcal{H}}{ \partial a^*_j}(\bar{a},\bar{a}^*)= \frac{\partial }{\partial a^*_j }\frac{\partial \mathcal{H}}{ \partial a^*_\ell}(\bar{a},\bar{a}^*)=g\frac{\partial \varphi_\ell}{\partial a^*_j}(\bar{a},\bar{a}^*).
\end{equation}

\begin{eg}
\label{eg:self-cross-Kerr}
    For instance, in the case of $N$ modes with self-Kerr nonlinearity of strength $g$ we have
    \begin{equation}
        \hat{\mathcal{H}}(\hat{a},\hat{a}^\dagger)= \sum_{j=1}^N
        \sum_{\ell=1}^N J_{j,\ell} \hat{a}_j^\dagger \hat{a}_\ell
        +\frac{g}{2}\sum_{j=1}^N \hat{a}_j^\dagger\hat{a}_j^\dagger \hat{a}_j\hat{a}_j,
    \end{equation}
    so, in the classical limit, the nonlinear terms arising in the linearization of the dynamical equations read 
    \begin{equation}
        \frac{\partial \varphi_j}{\partial a_\ell}(\bar{a},\bar{a}^*) = 2 |\bar{a}_j|^2 \, \delta_{j,\ell}, \quad \quad \frac{\partial \varphi_j}{\partial a^*_\ell}(\bar{a},\bar{a}^*) =   \bar{a}_j^2 \, \delta_{j,\ell},
    \end{equation}
    where $\delta_{j,\ell}$ indicates the Kronecker delta. Instead, in the case of $N$ modes with cross-Kerr nonlinearity of strength $g$ the Hamiltonian is
    \begin{equation}
    \hat{\mathcal{H}}(\hat{a},\hat{a}^\dagger)= 
    \sum_{j=1}^N \sum_{\ell=1}^N  J_{j,\ell} \hat{a}_j^\dagger \hat{a}_\ell
    +\frac{g}{2}\sum_{j\not = \ell} \hat{a}_j^\dagger 
 \hat{a}_j\hat{a}_\ell^\dagger  \hat{a}_\ell,
\end{equation}
and so in the classical limit we have for each $j \not = \ell$
    \begin{equation}
        \frac{\partial \varphi_j}{\partial a_\ell}(\bar{a},\bar{a}^*) =  \bar{a}_j \bar{a}^*_\ell, \quad \quad \frac{\partial \varphi_j}{\partial a^*_\ell}(\bar{a},\bar{a}^*) = \bar{a}_j \bar{a}_\ell,
    \end{equation}
    that are entries of, respectively, a Hermitian and a symmetric matrix.
\end{eg}

\begin{rem}
\label{rem:bog_transf_group}
The linearized scattering matrix $S(\bar{a},\bar{a}^*)$ is also a Bogoliubov transformation. In fact, assuming $H$ and $\frac{\partial \varphi}{\partial a^*}(\bar{a},\bar{a}^*)$ are invertible matrices, using \eqref{eqn:jacobian} and the block-matrix inverse formula we can write 
\begin{equation}
    W \coloneqq \nabla_{(a,a^*)}F(\bar{a},\bar{a}^*)^{-1}=\begin{pmatrix}
        W_1 & W_2\\
        W_2^* & W_1^*
    \end{pmatrix}
\end{equation}
where $W_1$ and $W_2$ are $N\times N$ matrices 
\begin{equation}
    W_1 = \biggl[\biggl( -iH-ig\frac{\partial \varphi}{\partial a}(\bar{a},\bar{a}^*) \biggl) - g^2 \biggl(\frac{\partial \varphi}{\partial a^*}(\bar{a},\bar{a}^*) \biggr)\biggl( iH^*+ig \biggl(\frac{\partial \varphi}{\partial a}(\bar{a},\bar{a}^*)\biggr)^* \biggl)^{-1}\biggl(\frac{\partial \varphi}{\partial a^*}(\bar{a},\bar{a}^*) \biggr)^*\biggr]^{-1}
\end{equation}
and
\begin{equation}
     W_2 = \biggl[ig \biggl(\frac{\partial \varphi}{\partial a^*}(\bar{a},\bar{a}^*) \biggr)^* -  \biggl( iH^*+ig\biggl(\frac{\partial \varphi}{\partial a}(\bar{a},\bar{a}^*)\biggl)^* \biggl)\biggl(- ig \frac{\partial \varphi}{\partial a^*}(\bar{a},\bar{a}^*) \biggr)^{-1}\biggl( -iH-ig\frac{\partial \varphi}{\partial a}(\bar{a},\bar{a}^*) \biggl)\biggr]^{-1}.
\end{equation}
Therefore $W$ and $S(\bar{a},\bar{a}^*)$ have the form of a Bogoliubov transformation.
\end{rem}

\begin{rem}[Reciprocity]
\label{rem:reciprocity}
Note that, in the linear case $(g=0),$ the linearized scattering matrix $S=\mathbf{I}+\sqrt{\kappa}M^{-1}\sqrt{\kappa}$ no longer depends on the steady state coordinates and has the following symmetries:
    \begin{equation}
    \label{eqn:qs1}
        S^\dagger = \sigma_y S \sigma_y,
    \end{equation}
    and 
    \begin{equation}
        \label{eqn:qs2}
        S^\dagger = \sigma_x S \sigma_x,
    \end{equation}
meaning that the system is reciprocal, i.e. light equally scatters in both directions between two different nodes $j$ and $\ell.$
In fact, if $g=0$:
\begin{equation}
    \begin{pmatrix}
        a_{\mathrm{out}}\\
        a_{\mathrm{out}}^*
    \end{pmatrix}=\underbrace{
    \begin{pmatrix}
       \Tilde{S} & 0\\
        0 & \Tilde{S}^*
    \end{pmatrix} 
    }_S    \begin{pmatrix}
        a_{\mathrm{in}}\\
        a_{\mathrm{in}}^*
    \end{pmatrix},
\end{equation}
where $\Tilde{S}= \mathbf{I}+i\sqrt{\kappa}H^{-1}\sqrt{\kappa}$ is the S matrix that one usually defines in the linear case at the zero frequency, and reciprocity, i.e. $\Tilde{S}=\Tilde{S}^\mathsf{T}$, implies any of \eqref{eqn:qs1} and \eqref{eqn:qs2} and vice-versa. In other words, the two symmetries above both capture the physical notion of reciprocity and are well posed in a nonlinear regime, where one has to consider both $a$ and $a^*$ in linearization.

\end{rem}
\begin{rem}[Non-holomorphic nonlinearity]
\label{rem:nonholomorphic}
    Indeed, if $\varphi(a,a^*)$ is an holomorphic function, i.e. if the Cauchy-Riemann equations $\frac{\partial \varphi}{\partial a^*}(a,a^*)=0$ hold, then the system would be reciprocal. Nevertheless, one can show that this is not a physical scenario, meaning that every optical nonlinearity also depends on $a^*.$ This follows from having monomial terms in the Hamiltonian containing at least two creation operators, which lead to monomials containing complex conjugate terms when computing \eqref{eqn:Poisson_b} (e.g. $\hat a_j^\dagger \hat a_j^\dagger \hat a_j \hat a_j$ leading to $2\, a_j^* a_j^2$).
\end{rem}
Therefore, if the system is nonlinear, such symmetry (reciprocity) is broken, nevertheless, for small values of $g,$ we can still observe little deviations from reciprocity. More rigorously:
\begin{prop}
\label{prop:quasi-rec}
    The linearized scattering matrix $S(\bar{a},\bar{a}^*) = \mathbf{I} + \sqrt{\kappa}[\nabla_{(a,a^*)}F(\bar{a},\bar{a}^*)]^{-1}\sqrt{\kappa}$ has the following quasi-symmetry 
    \begin{equation}
    \label{eqn:quasi-symm}
        S^{\dagger}(\bar{a},\bar{a}^*) = \sigma_y S(\bar{a},\bar{a}^*) \sigma_y + \mathcal{O}(g).
    \end{equation}
\end{prop}
\begin{proof}
    The matrix $M \coloneqq \nabla_{(a,a^*)}F(\bar{a},\bar{a}^*)$ can be rearranged as
    \begin{equation}
        M = \begin{pmatrix}
            -iH & 0\\
            0 & i H^*
        \end{pmatrix} -ig \begin{pmatrix}
        \frac{\partial \varphi}{\partial a }(\bar{a},\bar{a}^*) & \frac{\partial \varphi}{\partial a^* }(\bar{a},\bar{a}^*)\\
        -\left(\frac{\partial \varphi}{\partial a^* }(\bar{a},\bar{a}^*)\right)^* & - \left( \frac{\partial \varphi}{\partial a }(\bar{a},\bar{a}^*)\right)^*
    \end{pmatrix}.
    \end{equation}
    Recalling that $H$ is a symmetric matrix, we have
   \begin{equation}
   \label{eq:Mdag}
        M^\dagger = \begin{pmatrix}
            iH^* & 0\\
            0 & -i H
        \end{pmatrix} +ig \begin{pmatrix}
        \frac{\partial \varphi}{\partial a }(\bar{a},\bar{a}^*) & \frac{\partial \varphi}{\partial a^* }(\bar{a},\bar{a}^*)\\
        -\left(\frac{\partial \varphi}{\partial a^* }(\bar{a},\bar{a}^*)\right)^* & - \left( \frac{\partial \varphi}{\partial a }(\bar{a},\bar{a}^*)\right)^*
    \end{pmatrix}^\dagger,
    \end{equation}
    and
    \begin{equation}
    \label{eq:syMsy}
        \sigma_y M \sigma_y = \begin{pmatrix}
            iH^* & 0\\
            0 & -i H
        \end{pmatrix}-ig\begin{pmatrix}
        - \left( \frac{\partial \varphi}{\partial a }(\bar{a},\bar{a}^*)\right)^* & \left(\frac{\partial \varphi}{\partial a^* }(\bar{a},\bar{a}^*)\right)^*\\
         -\frac{\partial \varphi}{\partial a^* }(\bar{a},\bar{a}^*)& 
        \frac{\partial \varphi}{\partial a }(\bar{a},\bar{a}^*) 
    \end{pmatrix}.
    \end{equation}
From \eqref{eq:Mdag} and \eqref{eq:syMsy}, we can also write $M^\dagger = \Xi - g \Lambda$ and $ \sigma_y M \sigma_y=\Xi - g \Theta$, and using the Woodbory matrix identity we conclude
\begin{align}
(M^{-1})^\dagger -  \sigma_y M^{-1} \sigma_y &= (M^\dagger)^{-1} - (\sigma_y M \sigma_y)^{-1}= \sum_{m=0}^\infty \left(g \Xi^{-1}\Lambda \right)^m \Xi^{-1} -\sum_{m=0}^\infty \left(g \Xi^{-1}\Theta \right)^m \Xi^{-1}  \\
&= \sum_{m=1}^\infty g^m \left( (\Xi^{-1}\Lambda)^m - (\Xi^{-1}\Theta)^m \right)\Xi^{-1}=\mathcal{O}(g),
\label{eq:bound1}
\end{align}
where the first two power series above converge if the spectral radius of $g\Xi^{-1}\Lambda$ and $g\Xi^{-1}\Theta$ are less than one, which is true for $g$ small enough as we assume $H$ to be non-singular. From \eqref{eq:bound1} we can write
\begin{align}
    S(\bar{a},\bar{a}^*)^\dagger - \sigma_y S(\bar{a},\bar{a}^*)\sigma_y &= \Bigl( \mathbf{I}+\sqrt{\kappa}M^{-1}\sqrt{\kappa}\Bigr)^\dagger - \sigma_y\Bigl(\mathbf{I}+\sqrt{\kappa}M^{-1}\sqrt{\kappa}\Bigr)\sigma_y\\
    &= \Bigl( \mathbf{I}+\sqrt{\kappa}(M^{-1})^\dagger \sqrt{\kappa}\Bigr) - \Bigl(\mathbf{I}+\sqrt{\kappa}\sigma_y M^{-1}\sigma_y\sqrt{\kappa}\Bigr)\\
    &= \sqrt{\kappa}\sum_{m=1}^\infty g^m \left( (\Xi^{-1}\Lambda)^m - (\Xi^{-1}\Theta)^m \right)\Xi^{-1}\sqrt{\kappa}=\mathcal{O}(g)
\end{align}
using the fact that $\kappa$ is the diagonal matrix obtained repeating losses $\kappa_1, \dots \kappa_N$ twice and thus it commutes with $\sigma_y$.
Also, since
\begin{align}
    \| \Lambda \|, \| \Theta\| &\leq \left( \left\| \frac{\partial \varphi}{\partial a}(\bar{a},\bar{a}^*)\right\| + \left\| \frac{\partial \varphi^*}{\partial a}(\bar{a},\bar{a}^*)\right\|  + \left\| \frac{\partial \varphi}{\partial a^*}(\bar{a},\bar{a}^*)\right\| + \left\| \frac{\partial \varphi^*}{\partial a^*}(\bar{a},\bar{a}^*)\right\| \right)\\
    &=2\left( \left\| \frac{\partial \varphi}{\partial a}(\bar{a},\bar{a}^*)\right\| + \left\| \frac{\partial \varphi}{\partial a^*}(\bar{a},\bar{a}^*)\right\| \right),
\end{align}
for sub-multiplicative matrix norms we have
\begin{align}
    \|  S(\bar{a},\bar{a}^*)^\dagger - \sigma_y S(\bar{a},\bar{a}^*)\sigma_y \| &\leq \|\sqrt{\kappa}\|^2 \cdot \| \Xi^{-1} \| \sum_{m=1}^\infty g^m \Bigl( \|\Xi^{-1}\Lambda\|^m +  \|\Xi^{-1}\Theta\|^m \Bigr)\\
    &\leq \|\sqrt{\kappa}\|^2 \cdot \| \Xi^{-1} \| \sum_{m=1}^\infty g^m \|\Xi^{-1}\|^m\Bigl( \|\Lambda\|^m + \| \Theta\|^m \Bigr)\\
    &\leq 2 \|\sqrt{\kappa}\|^2 \cdot \| \Xi^{-1} \| \sum_{m=1}^\infty \Bigl(2g \| \Xi^{-1} \|\Bigr)^m \left( \left\| \frac{\partial \varphi}{\partial a}(\bar{a},\bar{a}^*)\right\| + \left\| \frac{\partial \varphi}{\partial a^*}(\bar{a},\bar{a}^*)\right\| \right)^m\\
    &= 4 \|\sqrt{\kappa}\|^2 \cdot \| \Xi^{-1} \| \frac{g\|\Xi^{-1}\| \left( \left\| \frac{\partial \varphi}{\partial a}(\bar{a},\bar{a}^*)\right\| + \left\| \frac{\partial \varphi}{\partial a^*}(\bar{a},\bar{a}^*)\right\| \right)}{1- 2g\|\Xi^{-1}\| \left( \left\| \frac{\partial \varphi}{\partial a}(\bar{a},\bar{a}^*)\right\| + \left\| \frac{\partial \varphi}{\partial a^*}(\bar{a},\bar{a}^*)\right\| \right)}.
\end{align}
In particular, for induced norms we have $\| \Xi^{-1} \|= \| H^{-1} \|$ and so
\begin{equation}
    \|  S(\bar{a},\bar{a}^*)^\dagger - \sigma_y S(\bar{a},\bar{a}^*)\sigma_y \| \leq 4 \|\sqrt{\kappa}\|^2 \cdot \| H^{-1} \| \frac{g\|H^{-1}\| \left( \left\| \frac{\partial \varphi}{\partial a}(\bar{a},\bar{a}^*)\right\| + \left\| \frac{\partial \varphi}{\partial a^*}(\bar{a},\bar{a}^*)\right\| \right)}{1- 2g\|H^{-1}\| \left( \left\| \frac{\partial \varphi}{\partial a}(\bar{a},\bar{a}^*)\right\| + \left\| \frac{\partial \varphi}{\partial a^*}(\bar{a},\bar{a}^*)\right\| \right)}.
\end{equation}
Whereas, if we consider the Frobenius norm
\begin{equation}
    \| \Xi^{-1} \|_F^2 = \Tr((\Xi^{-1})^\dagger \Xi^{-1})= 2 \| H^{-1} \|_F^2
\end{equation}
we have
\begin{equation}
    \|  S(\bar{a},\bar{a}^*)^\dagger - \sigma_y S(\bar{a},\bar{a}^*)\sigma_y \|_F \leq 8\|\sqrt{\kappa}\|_F^2 \cdot \| H^{-1} \|_F \frac{g\|H^{-1}\|_F \left( \left\| \frac{\partial \varphi}{\partial a}(\bar{a},\bar{a}^*)\right\|_F + \left\| \frac{\partial \varphi}{\partial a^*}(\bar{a},\bar{a}^*)\right\|_F \right)}{1- 2 \sqrt{2}g\|H^{-1}\|_F \left( \left\| \frac{\partial \varphi}{\partial a}(\bar{a},\bar{a}^*)\right\|_F + \left\| \frac{\partial \varphi}{\partial a^*}(\bar{a},\bar{a}^*)\right\|_F \right)}.
\end{equation}
As one would expect, in the case of self/cross-Kerr nonlinearities, assuming the steady state components $|\bar{a}_j|$ are of the same order of a reference input amplitude $|a_\mathrm{in}^0|$, then $\left\| \frac{\partial \varphi}{\partial a}(\bar{a},\bar{a}^*)\right\|$ and $\left\| \frac{\partial \varphi}{\partial a^*}(\bar{a},\bar{a}^*)\right\|$ in the equation above are quadratic in $|a_\mathrm{in}^0|$.
\end{proof}
Similarly, also the following other quasi-symmetry correspondent to \eqref{eqn:qs2} holds
\begin{prop}
\label{prop:quasi-rec-x}
    The linearized scattering matrix $S(\bar{a},\bar{a}^*) = \mathbf{I} + \sqrt{\kappa}[\nabla_{(a,a^*)}F(\bar{a},\bar{a}^*)]^{-1}\sqrt{\kappa}$ has the following quasi-symmetry 
    \begin{equation}
    \label{eqn:quasi-symm-x}
        S^{\dagger}(\bar{a},\bar{a}^*) = \sigma_x S(\bar{a},\bar{a}^*) \sigma_x + \mathcal{O}(g).
    \end{equation}
\end{prop}
\begin{proof}
    Follows from an analogue argument as above using
\begin{equation}
\label{eqn:sxWsx}
    \sigma_x M \sigma_x = \begin{pmatrix}
            iH^* & 0\\
            0 & -i H
        \end{pmatrix} -ig \begin{pmatrix}
         - \left( \frac{\partial \varphi}{\partial a }(\bar{a},\bar{a}^*)\right)^* &         -\left(\frac{\partial \varphi}{\partial a^* }(\bar{a},\bar{a}^*)\right)^* \\
        \frac{\partial \varphi}{\partial a^* }(\bar{a},\bar{a}^*)& 
        \frac{\partial \varphi}{\partial a }(\bar{a},\bar{a}^*)
    \end{pmatrix}.
\end{equation}
\end{proof}

\begin{rem}[Sparsity of nonlinearity and quasi-symmetry]
Recalling the above symmetries, combining \eqref{eq:Mdag} and \eqref{eq:syMsy} one obtains
\begin{equation}
    M^\dagger - \sigma_y M \sigma_y  = - 2 g \,\mathfrak{Im}\begin{pmatrix}
        \frac{\partial \varphi}{\partial a} (\bar{a},\bar{a}^*) & -\frac{\partial \varphi}{\partial a^*}(\bar{a},\bar{a}^*) \\
        -\frac{\partial \varphi}{\partial a^*}(\bar{a},\bar{a}^*) &\frac{\partial \varphi}{\partial a} (\bar{a},\bar{a}^*)
    \end{pmatrix},
\end{equation}
which for a network with only self-nonlinearities is sparse, as in this case $ \mathfrak{Im}\frac{\partial \varphi}{\partial a} (\bar{a},\bar{a}^*) = 0$ and $\frac{\partial \varphi}{\partial a^*}(\bar{a},\bar{a}^*)$ is diagonal (see Example \ref{eg:self-cross-Kerr} for self-Kerr). This, in practice, contributes in having a better gradient approximation as $N$ increases, as displayed in Fig.~\ref{fig:XOR}~\textbf{b}.
\end{rem}

\section{Gradient approximation for general systems with linear input-output relations}
\label{app:gradient-deriv}
In this section, we derive a general gradient approximation formula for applying Scattering Backpropagation to a system of differential equations
\begin{equation}
\label{eq:de1}
    \dot{\xi}=F_\theta(\xi)-\Pi \, \xi_\mathrm{in},
\end{equation}
with linear input-output relations
\begin{equation}
    \xi_\mathrm{out}=\Gamma \xi_\mathrm{in}+ \Sigma \xi,
\end{equation}
with invertible matrices $\Pi$ and $\Sigma$. In a supervised learning setting, we aim at efficiently estimating the derivative $\frac{\partial C}{\partial \theta}(\mathrm{y}, \mathrm{y_{target}})$ of the cost function $C(\mathrm{y}, \mathrm{y_{target}})$ with respect to the trainable parameters $\theta$. The latter, measures the deviation of the obtained neuromorphic output $\mathrm{y}$ (defined via the system output $\xi_\mathrm{out})$ from the expected target output $\mathrm{y_{target}}$ correspondent to a fixed input $\mathrm{x}$ encoded in $\xi_\mathrm{in}$.

As we will discuss, our gradient approximation depends on a quasi-symmetry of the (inverse of the) Jacobian of \eqref{eq:de1} at a steady state $\bar{\xi}$ (i.e. of the Green's function):
\begin{equation}
\label{eq:generalQS}
    (\nabla_\xi F_\theta (\bar{\xi})^{-1})^\dagger = U_1 \nabla_\xi F_\theta(\bar{\xi})^{-1}U_2+\mathcal{O}(g),
\end{equation}
where $U_1$ and $U_2$ are constant matrices and $g$ can be a non-trainable parameter of the system \eqref{eq:de1}. For instance, in our optical example \eqref{eqn:system2} $g$ is the nonlinearity strength and \eqref{eq:generalQS} is related to the system approximate reciprocity (broken by the optical nonlinearity). In this case, one usually consider either $U_1=U_2=\sigma_x$ or $U_1=U_2=\sigma_y.$ 

\begin{lem}
\label{lem:ift}
   In a system 
        $\dot{\xi} = F_\theta(\xi)-\Pi\, \xi_\mathrm{in}$ with input--output relations  $\xi_\mathrm{out} = \Gamma \xi_\mathrm{in} + \Sigma \,  \xi$, at the steady state $\bar{\xi}$ we have
\begin{equation}
    \frac{\partial \xi_\mathrm{out}}{\partial  \theta} (\bar{\xi}, \theta) = - \Sigma (\nabla_\xi F_\theta(\bar{\xi}))^{-1}\frac{\partial F_\theta}{\partial \theta}(\bar{\xi}).
\end{equation}
\end{lem}
\begin{proof}
The steady state equation reads $0=F_\theta(\bar{\xi})-\Pi\, \xi_\mathrm{in}$, under mild conditions on the regularity of $F_\theta$, by the implicit function theorem there exists a map $\theta \mapsto \bar{\xi}(\theta)$ that locally satisfies such equation. If we then differentiate the equation we get
\begin{equation}
    0=\frac{d}{d \theta}F_\theta(\bar{\xi}(\theta))=\frac{\partial F_\theta}{\partial \theta}(\bar{\xi}(\theta))+ \nabla_\xi F_\theta(\bar{\xi}(\theta)) \frac{\partial \bar{\xi}}{\partial \theta}(\theta).
\end{equation}
We conclude by solving for $\frac{\partial \bar{\xi}}{\partial \theta}$ and using the input--output relations.
\end{proof}
Note that in our optical case the last result reduces to
\begin{equation}
    \frac{\partial \xi_\mathrm{out}}{\partial  \theta} (\bar{\xi}, \theta) = (\mathbf{I} - S_\theta(\bar{\xi}, \theta) )\sqrt{\kappa^{-1}} \frac{\partial F_\theta}{\partial  \theta }(\bar{\xi}, \theta),
\end{equation}
where $S_\theta(\bar{\xi})=\mathbf{I} + \sqrt{\kappa} \,\nabla_\xi F_\theta(\bar{\xi})^{-1}\sqrt{\kappa}.$

\begin{theo}[Gradient approximation] 
\label{theo:gradient}
For a system evolving according to
$\dot{\xi} = F_\theta(\xi)-\Pi\, \xi_\mathrm{in}$ with linear input--output relations  $\xi_\mathrm{out} = \Gamma \xi_\mathrm{in} + \Sigma \,  \xi$ such that $\Pi$ and $\Sigma$ are invertible, the expression for the gradient of the cost function $C(\mathrm{y},\mathrm{y_{target}})$ in presence of a quasi-symmetry
\begin{equation}    
(\nabla_\xi F_\theta (\bar{\xi})^{-1})^\dagger = U_1 \nabla_\xi F_\theta(\bar{\xi})^{-1}U_2+\mathcal{O}(g),
\end{equation}
can be expressed as
\begin{align}
    \frac{\partial C}{\partial \theta}(\mathrm{y},\mathrm{y_{target}})
        &= - \left(\frac{\partial F_\theta}{\partial  \theta }(\bar{\xi})\right)^\dagger U_1 \Sigma^{-1} \; \frac{\delta \xi_\mathrm{out} - \Gamma \, \delta \xi_\mathrm{in}}{\beta} + \mathcal{O}(g,\beta),
\end{align}
where
\begin{equation}
\label{eq:general_error_signal}
    \delta \xi_\mathrm{in} \coloneqq \beta \Pi^{-1} U_2 \Sigma^\dagger \frac{\partial C}{\partial \xi^*_\mathrm{out}}(\mathrm{y},\mathrm{y}_\mathrm{target}).
\end{equation}
\end{theo}

\begin{proof}
    Considering the derivative w.r.t. a single parameter $\theta_j$ and applying Lemma \ref{lem:ift}, we can write
    \begin{align}
        \frac{\partial C}{\partial \theta_j} &= \left(  \frac{\partial C}{\partial \xi_\mathrm{out}} \right)^\mathsf{T} \frac{\partial\xi_\mathrm{out}}{\partial  \theta_j}\\
        &= -\left(  \frac{\partial C}{\partial \xi_\mathrm{out}} \right)^\mathsf{T}\Sigma (\nabla_\xi F_\theta(\bar{\xi}))^{-1}\frac{\partial F_\theta}{\partial \theta_j}(\bar{\xi}).
    \end{align}
We would like to interpret the last term $\frac{\partial F_\theta}{\partial  \theta_j}(\bar{\xi})$ as the injected error signal $\delta \xi_\mathrm{in},$ which would result in a measurable output change $\delta \xi_\mathrm{out}$, given by the action of the linearized scattering matrix onto $\delta \xi_\mathrm{in}.$ Nevertheless, this would be equivalent to training via parameter-shift, as it would require a number of experiments equal to the number of learnable parameters. In these cases, the main idea is to transpose the above expression so that the vector on the right no longer depends on $\theta$:
    \begin{align}
        \frac{\partial C}{\partial \theta_j} &= \left(  \frac{\partial C}{\partial \xi_\mathrm{out}} \right)^\mathsf{T} \frac{\partial\xi_\mathrm{out}}{\partial  \theta_j}= \left(\frac{\partial\xi_\mathrm{out}}{\partial  \theta_j}  \right)^\mathsf{T} \frac{\partial C}{\partial \xi_\mathrm{out}} =\left(\frac{\partial\xi_\mathrm{out}}{\partial  \theta_j}  \right)^\dagger \left( \frac{\partial C}{\partial \xi_\mathrm{out}}\right)^* \\
        \label{eqn:see_angle}
        &= -\left(\frac{\partial F_\theta}{\partial  \theta_j}(\bar{\xi})\right)^\dagger (\nabla_\xi F_\theta(\bar{\xi})^{-1})^\dagger \Sigma^\dagger \frac{\partial C}{\partial \xi^*_\mathrm{out}}= - \left(\frac{\partial F_\theta}{\partial  \theta_j}(\bar{\xi})\right)^\dagger  U_1(\nabla_\xi F_\theta(\bar{\xi})) ^{-1}\underbrace{U_2 \Sigma^\dagger \frac{\partial C}{\partial \xi^*_\mathrm{out}}}_{=: \Pi \, \delta \xi_\mathrm{in}/\beta} + \mathcal{O}(g)\\
        &= -\frac{1}{\beta}\left(\frac{\partial F_\theta}{\partial  \theta_j}(\bar{\xi})\right)^\dagger  U_1 \Sigma^{-1} \underbrace{\Sigma (\nabla_\xi F_\theta(\bar{\xi})) ^{-1}\Pi}_{S_\theta(\bar{\xi})-\Gamma} \delta \xi_\mathrm{in}+ \mathcal{O}(g) = -\left(\frac{\partial F_\theta}{\partial  \theta_j}(\bar{\xi})\right)^\dagger  U_1 \Sigma^{-1} \frac{\delta \xi_\mathrm{out} - \Gamma \delta \xi_\mathrm{in}}{\beta}+ \mathcal{O}(g)+ \mathcal{O}(\beta)
    \end{align}
    where we have respectively used the fact that $\frac{\partial C}{\partial \theta_j}$ is real, Lemma \ref{lem:ift}, Proposition \ref{prop:quasi-rec}, and Lemma \ref{lem:smatrix}. 
\end{proof}
Note that if the matrices $U_1$, $U_2$, $\Sigma,$ $\Gamma$, and $\Pi$ are local, than also the gradient approximation above proscribes local updates for the trainable parameters.
Furthermore, in the optical case modeled by \eqref{eqn:system2} and \eqref{eq:inout1} having $\Gamma = \mathbf{I}$ and $\Sigma=\Pi =\sqrt{\kappa}$, 
the gradient formula above takes the form
\begin{align}
\label{eq:ugly-grad1}
    \frac{\partial C}{\partial \theta}
    = 
    - \left(\frac{\partial F_\theta}{\partial  \theta }(\bar{\xi})\right)^\dagger U_1 \sqrt{\kappa^{-1}} \; \frac{\delta  \xi_\mathrm{out} - \delta  \xi_\mathrm{in}}{\beta} +  \mathcal{O}\left(g,\beta\right),
\end{align}
where 
\begin{equation}
\label{eq:ugly-grad2}
    \delta \xi_\mathrm{in} \coloneqq \beta \sqrt{\kappa^{-1}} U_2 \sqrt{\kappa}\frac{\partial C}{\partial \xi^*_\mathrm{out}}(\mathrm{y},\mathrm{y}_\mathrm{target}).
\end{equation}
Notice the latter are slightly different from the equations derived in the Methods section, as there we instead assumed (in favor of simplicity) the quasi-symmetry 
\begin{equation}
\label{eq:qs-scatt-M}
    S_\theta(\bar{\xi})^\dagger = U S_\theta(\bar{\xi}) U^{-1}+\mathcal{O}(g)
\end{equation}
in place of \eqref{eq:generalQS}, where $S_\theta(\bar{\xi})=\mathbf{I} + \sqrt{\kappa} \,\nabla_\xi F_\theta(\bar{\xi})^{-1}\sqrt{\kappa}.$ Nevertheless, in this case, \emph{if} the matrix $U = U_1 = U_2^{-1}$ and commutes with $\sqrt{\kappa}$ (which is true in our optical scenario for $U=\sigma_x$ or $U=\sigma_y$) then \eqref{eq:generalQS} implies \eqref{eq:qs-scatt-M} and equations \eqref{eq:ugly-grad1} and \eqref{eq:ugly-grad2} reduce to
\begin{align}
    \frac{\partial C}{\partial \theta}
    = 
    - \left(\frac{\partial F_\theta}{\partial  \theta }(\bar{\xi})\right)^\dagger \sqrt{\kappa^{-1}}U \; \frac{\delta  \xi_\mathrm{out} - \delta  \xi_\mathrm{in}}{\beta} +  \mathcal{O}\left(g,\beta\right),
\end{align}
where 
\begin{equation}
    \delta \xi_\mathrm{in} \coloneqq \beta U^{-1}\frac{\partial C}{\partial \xi^*_\mathrm{out}}(\mathrm{y},\mathrm{y}_\mathrm{target}).
\end{equation}
\section{Gradient approximation for quasi-reciprocal systems}
\begin{prop} 
\label{Prop:gradient}
For the system \eqref{eqn:system2}, it holds
\begin{align}
\label{eqn:det_grad_approx}
    \frac{\partial C}{\partial \Delta_j} = -\frac{2}{\kappa_j} \mathfrak{Re}\left[ (a_{\mathrm{out},j}-a_{\mathrm{in},j})\frac{\delta a_{\mathrm{out},j}- \delta a_{\mathrm{in},j}}{\beta} \right] + \mathcal{O}(g,\beta).
\end{align}
and
\begin{align}
   \frac{\partial C}{\partial J_{j,\ell}} = - &\frac{2}{\sqrt{\kappa_j \kappa_\ell}} \mathfrak{Re}\biggl[ (a_{\mathrm{out}\ell}-a_{\mathrm{in}\ell})\frac{\delta a_{\mathrm{out},j}- \delta a_{\mathrm{in},j}}{\beta}+ (a_{\mathrm{out},j}-a_{\mathrm{in},j})\frac{\delta a_{\mathrm{out}\ell}- \delta a_{\mathrm{in}\ell}}{\beta} \biggr]+\mathcal{O}(g,\beta). 
\label{eqn:coupl_grad_approx}
\end{align}
where 
\begin{equation}
\label{eqn:error_signal}
    \delta a_\mathrm{in} \coloneqq -i \beta \frac{\partial C}{\partial a_{\mathrm{out},j}}(\mathrm{y},\mathrm{y}_\mathrm{target}). 
\end{equation}
\end{prop}
\begin{proof}
    The claim is obtained by applying Theorem~\ref{theo:gradient} to the system \eqref{eqn:system2} choosing $U=U_1=U_2^{-1}$ as $\sigma_y\coloneqq \begin{pmatrix}
        0 & -i\mathbf{I}_N\\
        i\mathbf{I}_N & 0
    \end{pmatrix}$.
\end{proof}

In practice, in order to have an efficient training algorithm one has to pay the price of approximating the true gradients with the expression above. The error depends on the quasi-symmetry \eqref{eqn:quasi-symm} or, more precisely, on the angle between $\nabla_{(a,a^*)}F(\bar{a},\bar{a}^*)^{-1}$ and $\sigma_y\nabla_{(a,a^*)}F(\bar{a},\bar{a}^*)^{-1}\sigma_y$ (linked to the angle between $S(\bar{a},\bar{a}^*)^\dagger$ and $\sigma_y S(\bar{a},\bar{a}^*)\sigma_y$ showed in Figure \ref{fig:XOR}~\textbf{a} for fully connected networks with self-Kerr nonlinearities) which determines the angle between the true gradient and the approximated (see \eqref{eqn:see_angle}).

\begin{prop}
\label{prop:angle_W}
    Let $W \coloneqq \nabla_{(a,a^*)}F(\bar{a},\bar{a}^*)^{-1}$ be the inverse of the Jacobian of \eqref{eqn:system2} at the steady state and assume $H$ is invertible. Then
    \begin{equation}
      \alpha= \cos^{-1} \frac{\langle W^\dagger, \sigma_y W \sigma_y \rangle_F}{\|W^\dagger \|_F \cdot \| \sigma_y W \sigma_y \|_F} =\mathcal{O}(g).
    \end{equation}
\end{prop}
\begin{proof}
First, note that being $W$ a Bogoliubov transformation (see Remark \ref{rem:bog_transf_group}) implies $\langle W^\dagger, \sigma_y W \sigma_y \rangle_F=\Tr(W \sigma_y W \sigma_y)$ is a real number, and so 
\begin{equation}
    \alpha \coloneqq \cos^{-1} \frac{\mathfrak{Re} \langle W^\dagger, \sigma_y W \sigma_y \rangle_F}{\|W^\dagger \|_F \cdot \| \sigma_y W \sigma_y \|_F} = \cos^{-1} \frac{\langle W^\dagger, \sigma_y W \sigma_y \rangle_F}{\|W^\dagger \|_F \cdot \| \sigma_y W \sigma_y \|_F}.
\end{equation}
Since, for inner products, in general holds
\begin{equation}
    \langle C, D \rangle + \langle D, C \rangle = \langle C, C \rangle + \langle D, D \rangle - \langle C-D, C-D \rangle
\end{equation}
we have 
\begin{align}
    \cos \alpha = \frac{\langle W^\dagger,\sigma_y W \sigma_y \rangle_F}{\|W^\dagger\|_F \cdot \|\sigma_y W \sigma_y \|_F}&=\frac{\langle W^\dagger, W^\dagger \rangle_F+\langle \sigma_y W \sigma_y ,\sigma_y W \sigma_y \rangle_F - \langle W^\dagger-\sigma_y W \sigma_y ,W^\dagger-\sigma_y W \sigma_y \rangle_F }{2\|W^\dagger\|_F \cdot \|\sigma_y W \sigma_y \|_F}\\
    &=\frac{\langle W^\dagger, W^\dagger \rangle_F+\langle \sigma_y W \sigma_y ,\sigma_y W \sigma_y \rangle_F - \mathcal{O}(g^2)}{2\|W^\dagger\|_F \cdot \|\sigma_y W \sigma_y \|_F}\\
    &=\frac{\Tr (W^\dagger W )+\Tr(\sigma_y W^\dagger \sigma_y \sigma_y W \sigma_y )- \mathcal{O}(g^2)}{2\sqrt{\Tr (W^\dagger W )\Tr(\sigma_y W^\dagger \sigma_y \sigma_y W \sigma_y )}}\\
    &=\frac{2\Tr (W^\dagger W ) - \mathcal{O}(g^2)}{2\Tr (W^\dagger W )}= 1 - \mathcal{O}(g^2)
\end{align}
where we used quasi-reciprocity, i.e. Proposition \eqref{prop:quasi-rec}.
\end{proof}

\begin{prop}
\label{prop:angle_S}
    Let $S(\bar{a},\bar{a}^*)$ be the linerized scattering matrix of \eqref{eqn:system2} at the steady state and assume $H$ is invertible. Then
    \begin{equation}
      \alpha= \cos^{-1} \frac{\langle S^\dagger(\bar{a},\bar{a}^*), \sigma_y S(\bar{a},\bar{a}^*) \sigma_y \rangle_F}{\|S(\bar{a},\bar{a}^*)^\dagger \|_F \cdot \| \sigma_y S(\bar{a},\bar{a}^*) \sigma_y \|_F}=\mathcal{O}(g).
    \end{equation}
\end{prop}
\begin{proof}
    Analogue to the previous one.
\end{proof}

    Similar results on gradient and angle approximation can be shown with respect to the other quasi-symmetry we discussed in Proposition \ref{prop:quasi-rec-x}. For instance, one can show
\begin{prop}
For the system \eqref{eqn:system2}, it holds
  \begin{align}
  \label{eq:dc/dl_sigma_x}
    \frac{\partial C}{\partial \Delta_j} = -\frac{2}{\kappa_j} \mathfrak{Im}\left[ (a_{\mathrm{out},j}-a_{\mathrm{in},j})\frac{\delta a_{\mathrm{out},j}- \delta a_{\mathrm{in},j}}{\beta} \right] + \mathcal{O}(g,\beta)
\end{align}
and
\begin{equation}
  \label{eq:dc/dJ_sigma_x}
   \frac{\partial C}{\partial J_{j,\ell}} = -\frac{2}{\sqrt{\kappa_j \kappa_\ell}} \mathfrak{Im}\left[ (a_{\mathrm{out}\ell}-a_{\mathrm{in}\ell})\frac{\delta a_{\mathrm{out},j}- \delta a_{\mathrm{in},j}}{\beta} + (a_{\mathrm{out},j}-a_{\mathrm{in},j})\frac{\delta a_{\mathrm{out}\ell}- \delta a_{\mathrm{in}\ell}}{\beta} \right] + \mathcal{O}(g,\beta)
\end{equation}
where
\begin{equation}
    \delta a_{\mathrm{in}} \coloneqq  \beta \frac{\partial C}{\partial a_{\mathrm{out},j}}.
\end{equation}  
\end{prop}
\begin{proof}
    The claim is obtained by applying Theorem~\ref{theo:gradient} to the system \eqref{eqn:system2} choosing $U=U_1=U_2^{-1}$ as $\sigma_x\coloneqq \begin{pmatrix}
        0 & \mathbf{I}_N\\
        \mathbf{I}_N & 0
    \end{pmatrix}$.
\end{proof}

\section{Quadrature basis}
In the $(x,p)$-quadrature basis
\begin{equation}
    x_j \coloneqq \frac{a_j + a^*_j}{\sqrt{2}}, \quad p_j \coloneqq i\frac{a^*_j - a_j}{\sqrt{2}}
\end{equation}
the dynamical equations \eqref{eqn:system2} read
\begin{equation}
    \label{eqn:quadratures_Langevin}
    \begin{cases}
        \dot{x} = -\frac{\kappa+\kappa'}{2}x + Jp +g \, \mathfrak{Im} \varphi\left(\frac{x+ip}{\sqrt{2}},\frac{x-ip}{\sqrt{2}}\right)-\sqrt{\kappa_j}x_{\mathrm{in}}\\
        \dot{p} = -\frac{\kappa+\kappa'}{2}p - Jx -g \, \mathfrak{Re} \varphi\left(\frac{x+ip}{\sqrt{2}},\frac{x-ip}{\sqrt{2}}\right)-\sqrt{\kappa_j}p_{\mathrm{in}}
    \end{cases}
\end{equation}
The Jacobian matrix of this ODE at a steady state $(\bar{x},\bar{p})$ is
\begin{equation}
\label{eqn:Jac_xp}
    \nabla F_{xp}(\bar{x},\bar{p}) =\begin{pmatrix}
        -\frac{\kappa+\kappa'}{2} + g \frac{\partial}{\partial x}\mathfrak{Im} \varphi\left(\frac{\bar x+i\bar p}{\sqrt{2}},\frac{\bar x-i\bar p}{\sqrt{2}}\right) & J + g \frac{\partial}{\partial p}\mathfrak{Im} \varphi\left(\frac{\bar x+i\bar p}{\sqrt{2}},\frac{\bar x-i\bar p}{\sqrt{2}}\right)\\
        -J - g \frac{\partial}{\partial x}\mathfrak{Re} \varphi\left(\frac{\bar x+i\bar p}{\sqrt{2}},\frac{\bar x-i\bar p}{\sqrt{2}}\right) & -\frac{\kappa+\kappa'}{2} - g \frac{\partial}{\partial p}\mathfrak{Re} \varphi\left(\frac{\bar x+i\bar p}{\sqrt{2}},\frac{\bar x-i\bar p}{\sqrt{2}}\right)
    \end{pmatrix}
\end{equation}
and, we the usual abuse of notation on the loss matrix, we can define a linearized scattering matrix in this basis 
\begin{equation}
    S_{xp}(\bar{x},\bar{p}) \coloneqq \mathbf{I}_{2N} + \sqrt{\kappa}\nabla F_{xp}(\bar{x},\bar{p})^{-1}\sqrt{\kappa}.
\end{equation}
It is easy to show that for $S_{xp}$, in the linear case $g=0$, hold symmetries
\begin{equation}
\label{eqn:qsx}
    S_{xp}^\mathsf{T} = \sigma_x S_{xp}\sigma_x, \quad \sigma_x \coloneqq \begin{pmatrix}
        0 & \mathbf{I}_N \\
        \mathbf{I}_N & 0
    \end{pmatrix}
\end{equation}
and
\begin{equation}
\label{eqn:qsz}
    S_{xp}^\mathsf{T} = \sigma_z S_{xp}\sigma_z, \quad \sigma_z \coloneqq \begin{pmatrix}
        \mathbf{I}_N & 0\\
        0 & -\mathbf{I}_N
    \end{pmatrix}
\end{equation}
correspondent to \eqref{eqn:qs1} and \eqref{eqn:qs2}, which are equivalent to system's reciprocity. In the nonlinear regime, by a change of basis also follow the results on the respective quasi-symmetries that we discussed above in Propositions \ref{prop:quasi-rec}, \ref{prop:quasi-rec-x}, \ref{prop:angle_W}.

\section{Comparison with Equilibrium Propagation for vector fields}
\label{app:EP}

\subsection{Equilibrium Propagation for vector fields}
In \cite{scellier2018generalization}, authors generalize Equilibrium Propagation (EP) for fixed point systems whose dynamics is described by vector fields, so extending the method also to non-energy-based models. In practice, in a supervised setting aiming to predict the target $\mathrm{y}$ of a network input $\mathrm{x}$,
they consider the evolution of neurons $s$ described by
\begin{equation}
    \frac{ds}{dt}=\mu_\theta(\mathrm{x},s)
\end{equation}
assuming that the system evolves towards a fixed point $s_\theta^{\mathrm{x}}$ depending on $\mathrm{x}$ and the system's parameters $\theta$ via the implicit relation
\begin{equation}
\mu_\theta(\mathrm{x},s_\theta^{\mathrm{x}})=0.
\end{equation}
In particular, they consider the example of a model with two hidden layers ($s_1$ and $s_2$) and one output layers ($s_0$) and a component-wise defined vector field
\begin{align}
    \mu_{\theta,0}(\mathrm{x},s) &= W_{01}  \rho(s_1) - s_0\\
    \mu_{\theta,1}(\mathrm{x},s) &= W_{12}  \rho(s_1)+W_{01}  \rho(s_1) - s_1\\
    \mu_{\theta,2}(\mathrm{x},s) &= W_{23}  \rho(s_1)+W_{21} \rho(s_1) - s_2 
\end{align}
where the $W$s indicate the trainable weight matrices. Here, unlike energy-based models like Hopfield \cite{hopfield1984neurons} which assume symmetric coupling among neurons, the tunable connections are not tied.
In this example, authors consider a quadratic cost function depending on the output layer $s_0$
\begin{equation}
    C(\mathrm{y},s)=\frac{1}{2}\| \mathrm{y} - s_0 \|^2
\end{equation}
and aim to minimize $J(\mathrm{x},\mathrm{y},\theta)\coloneqq C(\mathrm{y},s_\theta^\mathrm{x})$ with respect to $\theta$, by proposing an algorithm to compute an approximation of
\begin{equation}
    \frac{\partial J}{\partial \theta}(\mathrm{x},\mathrm{y},\theta),
\end{equation}
whose precision depends on the `degree of symmetry' of the Jacobian matrix of $\mu_\theta$ at the fixed point $s_\theta^\mathrm{x}$.

Inspired by the original version of EP for energy-based models \cite{scellier2017equilibrium} where the cost function $C$ is seen as an `external potential energy' in the output later $s_0$ that drives the network prediction towards the target $\mathrm{y}$, they define an \emph{augmented vector field}
\begin{equation}
\label{eqn:augmented_vf}
    \mu^\beta_\theta(\mathrm{x},\mathrm{y},s) \coloneqq \mu_\theta(\mathrm{x},s) - \beta \frac{\partial C}{\partial s}(\mathrm{y},s),
\end{equation}
where $\beta \geq 0$ is the \emph{influence parameter}. Thus, in this augmented field, the term $\frac{\partial C}{\partial s}(\mathrm{y},s)$ can be viewed as an \emph{external force} that nudges the network output towards the target. The main idea of this generalized EP, is to apply the update rule
\begin{equation}
    \Delta \theta \propto \nu(\mathrm{x},\mathrm{y},\theta),
\end{equation}
where
\begin{equation}
    \nu(\mathrm{x},\mathrm{y},\theta) \coloneqq \left( \frac{\partial \mu_\theta}{\partial \theta}(\mathrm{x},s_\theta^\mathrm{x}) \right)^\mathsf{T} \, \frac{\partial s_\theta^\beta}{\partial \beta} \biggl\vert_{\beta = 0}
\end{equation}
can be computed with two `experiments' in a free phase, when $\beta = 0$ and we measure $s_\theta^0 \coloneqq s_\theta^\mathrm{x}$, and a nudged phase, when $\beta > 0$ and one measures $s_\theta^\beta$ defined by
\begin{equation}
\mu^\beta_\theta(\mathrm{x},\mathrm{y},s_\theta^\beta) = 0.
\end{equation}
It turns out, see Theorem 1 in \cite{scellier2018generalization}, that the angle between the vectors $\frac{\partial J}{\partial \theta}(\mathrm{x},\mathrm{y},\theta)$ and $-\nu(\mathrm{x},\mathrm{y},\theta)$ is linked to the `degree of symmetry' of the Jacobian matrix of $\mu_\theta$ at the fixed point $s_\theta^\mathrm{x}$, leading to an exact computation of the gradient in the case of symmetric weights (in accordance with energy-based EP).

\subsection{Scattering Backpropagation as generalization of Equilibrium Propagation for vector fields}
In our case, we start from a real vector field in the $(x,p)$-quadratures of the form \eqref{eqn:quadratures_Langevin}, which we report here for convenience
\begin{equation}
\label{eqn:quadratures_Langevin_pt2}
    \begin{cases}
        \dot{x} = -\frac{\kappa+\kappa'}{2}x + Jp +g \, \mathfrak{Im} \varphi\left(\frac{x+ip}{\sqrt{2}},\frac{x-ip}{\sqrt{2}}\right)-\sqrt{\kappa_j}x_{\mathrm{in}}\\
        \dot{p} = -\frac{\kappa+\kappa'}{2}p - Jx -g \, \mathfrak{Re} \varphi\left(\frac{x+ip}{\sqrt{2}},\frac{x-ip}{\sqrt{2}}\right)-\sqrt{\kappa_j}p_{\mathrm{in}}
    \end{cases}
\end{equation}
At this point, one could think to directly apply EP for vector fields to these dynamical equations, nevertheless, there are two main obstructions
\begin{itemize}
    \item The Jacobian matrix of \eqref{eqn:quadratures_Langevin_pt2} at the steady state $(\bar{x}, \bar{p})$ is very far from being symmetric (see \eqref{eqn:Jac_xp}), even for small values of $g$ and with symmetric couplings $J$ (required by the physics). Therefore, the approximation given by $\nu(\mathrm{x},\mathrm{y},\theta)$ is not accurate and the optimization method does not converge.
    \item Defining an augmented vector field can be practically difficult in the optical implementations we aim at. Indeed, defining \eqref{eqn:augmented_vf} would imply being able to modify the Hamiltonian of our system, thing which is often impossible to engineer for many choices of cost function $C$, e.g. for the cross-entropy loss we chose for training the CNN-like setup.
\end{itemize}

Notably, system \eqref{eqn:quadratures_Langevin_pt2} is an example of the well-studied Port-Hamiltonian systems \cite{van2014port}, a class of differential equations which models Hamiltonian systems in presence of dissipation an input-output relations. Thus, the training method we propose to address above points can be applied to more general dynamical systems even outside the optical domain (see Example \ref{eg:PortHS}).

In the following, we will briefly show how the training method we exposed in the main text, for systems with input-output relations $s_\mathrm{out}=s(t),$ can be viewed as a generalization of vector field EP to solve these two obstructions. In accordance with the previous section, we will keep the same notations for a general system of neurons $s$ following the vector field dynamics
\begin{equation}
\label{eqn:vector_field_dyn}
\frac{ds}{dt}=\mu_\theta(\mathrm{x},s),
\end{equation}
and reaching a steady state (fixed point) $s_\theta^0$
\begin{equation}
    \mu_\theta(\mathrm{x},s_\theta^0)=0.
\end{equation}
We will show that, if 
\begin{equation}
    \nabla \mu_\theta(\mathrm{x},s_\theta^0)^{-1}\approx P_1 \, (\nabla \mu_\theta(\mathrm{x},s_\theta^0)^{-1})^\mathsf{T} P_2,
\end{equation} for some constant matrices $P_1$ and $P_2$, in the sense that the angle (defined with respect to some scalar product) between them is small, then it is possible to approximate the gradient $\frac{\partial J}{\partial \theta }(\mathrm{x},\mathrm{y},\theta)$ similarly as we discussed above.

\begin{rem}
In the original EP for vector fields one would have $P_1=P_2=\mathbf{I},$ whereas in our optical example \eqref{eqn:quadratures_Langevin_pt2} one can choose e.g. $P_1=P_2=\sigma_x$ or $P_1=P_2=\sigma_z$ (see equations \eqref{eqn:qsx} and \eqref{eqn:qsz}).
\end{rem}

As always, the training scheme consists of two phases. In the first one, we let the system evolve towards a steady state $s_\theta^0$ that we measure and consider as the neuromorphic architecture's output, allowing us to compute  $C(\mathrm{y},s_\theta^0)$ and $\frac{\partial C}{\partial s}(\mathrm{y},s_\theta^0)$. Then, we slightly modify the system defining the \emph{augmented vector field}  
\begin{equation}
\label{eqn:augmented_vf_new}
    \mu^\beta_\theta(\mathrm{x},\mathrm{y},s) \coloneqq \mu_\theta(\mathrm{x},s) - \beta P^\mathsf{T}_1 \frac{\partial C}{\partial s}(\mathrm{y},s_\theta^0),
\end{equation}
which differs from \eqref{eqn:augmented_vf} for the presence of $P_1^\mathsf{T}$ and the fact that the derivative is evaluated at the first fixed point $s_\theta^0$. This is crucial as it allows us to interpret such term as a perturbation of the external probe field, i.e. $\delta a_\mathrm{in}$, as we did in the main text. In this way, we do not need to engineer a different Hamiltonian for this nudged system, but only to inject a small error signal on top of our original input.
In the second phase, we let evolve the nudged system 
\begin{equation}
    \frac{ds}{dt}=\mu^\beta_\theta(\mathrm{x},\mathrm{y},s)
\end{equation}
towards a new nudged equilibrium $s_\theta^\beta$ defined by
\begin{equation}
\mu^\beta_\theta(\mathrm{x},\mathrm{y},s_\theta^\beta) = 0.
\end{equation}
Finally, we can apply the update rule
\begin{equation}
    \Delta \theta \propto \nu(\mathrm{x},\mathrm{y},\theta),
\end{equation}
where now we define
\begin{equation}
\label{eqn:nu_new}
    \nu(\mathrm{x},\mathrm{y},\theta) \coloneqq \left( P_2 \, \frac{\partial \mu_\theta}{\partial \theta}(\mathrm{x},s_\theta^\mathrm{x}) \right)^\mathsf{T}  \, \frac{\partial s_\theta^\beta}{\partial \beta} \biggl\vert_{\beta = 0},
\end{equation}
which can be again be computed with the two fixed points measured in the two phases, approximating
\begin{equation}
    \frac{\partial s_\theta^\beta}{\partial \beta} \biggl\vert_{\beta = 0} = \frac{s_\theta^\beta - s_\theta^0}{\beta} + \mathcal{O}(\beta).
\end{equation}

As we show in the following result, it turns out that the angle between the gradient approximation $\left( \frac{\partial J}{\partial \theta } \right)_{\mathrm{approx}} \coloneqq -\nu(\mathrm{x},\mathrm{y},\theta)^\mathsf{T}$ and the true gradient $\frac{\partial J}{\partial \theta }$ depends on the angle between the inverse of Jacobian $\nabla \mu_\theta(\mathrm{x},s_\theta^0)^{-1}$ and the matrix $P_1 \, (\nabla \mu_\theta(\mathrm{x},s_\theta^0)^{-1})^\mathsf{T} P_2$.

\begin{theo}
\label{theo:gen_EP}
    The true gradient $\frac{\partial J}{\partial \theta}(\mathrm{x},\mathrm{y},\theta)$ and the approximation $\left( \frac{\partial J}{\partial \theta } \right)_{\mathrm{approx}} \coloneqq -\nu(\mathrm{x},\mathrm{y},\theta)^\mathsf{T}$ can be expressed as
    \begin{equation}
         \frac{\partial J}{\partial \theta}(\mathrm{x},\mathrm{y},\theta) = - \left( \frac{\partial C}{\partial s}(\mathrm{y},s_\theta^0) \right)^\mathsf{T} \biggl( \nabla \mu_\theta(\mathrm{x},s_\theta^0) \biggr)^{-1} \, \frac{\partial \mu_\theta}{\partial \theta}(\mathrm{x},s_\theta^0)
    \end{equation}
    \begin{equation}
        \left( \frac{\partial J}{\partial \theta } \right)_{\mathrm{approx}}(\mathrm{x},\mathrm{y},\theta)= - \left( \frac{\partial C}{\partial s}(\mathrm{y},s_\theta^0) \right)^\mathsf{T}  P_1\biggl( \nabla \mu_\theta(\mathrm{x},s_\theta^0)^{-1}\biggr)^{\mathsf{T} }P_2\frac{\partial \mu_\theta}{\partial \theta}(\mathrm{x},s_\theta^0)
    \end{equation}
\end{theo}
\begin{proof}
    The main idea of the proof is the same of Theorem 1 in \cite{scellier2018generalization}, i.e. to use the implicit function theorem and differentiate the fixed point equation
    \begin{equation}
\mu^\beta_\theta(\mathrm{x},\mathrm{y},s_\theta^\beta) = 0
\end{equation}
with respect to $\theta$ and $\beta$ to compute $\frac{\partial s_\theta^\beta}{\partial \theta}$ and $\frac{\partial s_\theta^\beta}{\partial \beta}$. We find respectively
\begin{equation}
\label{eqn:dsdtheta}
   \frac{\partial s_\theta^\beta}{\partial \theta} = - \biggl( \nabla \mu^\beta_\theta(\mathrm{x},s_\theta^\beta) \biggr)^{-1} \, \frac{\partial \mu^\beta_\theta}{\partial \theta}(\mathrm{x},s_\theta^\beta)
\end{equation}
\begin{equation}
\label{eqn:nu}
       \frac{\partial s_\theta^\beta}{\partial \beta} = - \biggl( \nabla \mu^\beta_\theta(\mathrm{x},s_\theta^\beta) \biggr)^{-1} \, \frac{\partial \mu^\beta_\theta}{\partial \beta}(\mathrm{x},s_\theta^\beta) = \biggl( \nabla \mu^\beta_\theta(\mathrm{x},s_\theta^\beta) \biggr)^{-1} \, P_1^\mathsf{T}\frac{\partial C}{\partial s}(\mathrm{y},s_\theta^0),
\end{equation}
where we used that $ \frac{\partial \mu^\beta_\theta}{\partial \beta} = -P_1^\mathsf{T} \frac{\partial C}{\partial s}(\mathrm{y},s_\theta^0)$ by definition.
The latter, substituted into \eqref{eqn:nu_new} gives
\begin{align}
        \nu(\mathrm{x},\mathrm{y},\theta) &\coloneqq \left( P_2 \, \frac{\partial \mu_\theta}{\partial \theta}(\mathrm{x},s_\theta^\mathrm{x}) \right)^\mathsf{T}  \, \frac{\partial s_\theta^\beta}{\partial \beta} \biggl\vert_{\beta = 0}\\
        &= \left( P_2 \, \frac{\partial \mu_\theta}{\partial \theta}(\mathrm{x},s_\theta^\mathrm{x}) \right)^\mathsf{T}  \, \biggl( \nabla \mu_\theta(\mathrm{x},s_\theta^0) \biggr)^{-1} \, P_1^\mathsf{T} \frac{\partial C}{\partial s}(\mathrm{y},s_\theta^0).
\end{align}
Finally, we conclude by inserting \eqref{eqn:dsdtheta} into
\begin{equation}
    \frac{\partial J}{\partial \theta}(\mathrm{x},\mathrm{y},\theta) = - \left( \frac{\partial C}{\partial s}(\mathrm{y},s_\theta^0) \right)^\mathsf{T}\frac{\partial s_\theta^0}{\partial \theta} =- \left( \frac{\partial C}{\partial s}(\mathrm{y},s_\theta^0) \right)^\mathsf{T} \frac{\partial s_\theta^\beta}{\partial \theta} \biggl\vert_{\beta = 0}.
\end{equation}
\end{proof}
This result can be seen as a generalization of Theorem 1 in \cite{scellier2018generalization} in the presence of a more general `quasi-symmetry' and with a constant nudge (error signal) $\epsilon \coloneqq - \beta P_1^\mathsf{T} \frac{\partial C}{\partial s}(\mathrm{y},s_\theta^0)$ in the augmented vector field \eqref{eqn:augmented_vf_new}. These changes allow us to overcome the two obstructions we presented at the beginning of the section and consist in the main differences between our scheme and EP for training general fixed-point vector fields dynamics.

\begin{eg}
\label{eg:PortHS}
Scattering Backpropagation, can be applied to train general fixed point-dynamical systems by experimentally extracting the exact gradient $\partial_\theta C$ (or an approximation, depending on the system's symmetries). For instance, a large subclass of Port-Hamiltonian systems can be described by
\begin{equation}
\begin{cases}
       \dot{s}=\mu(s)\coloneqq J \, \nabla H(s) - R(s,\nabla H(s))+Gu,\\
    y = G^\mathsf{T}\nabla H(s), 
\end{cases}
\end{equation}
in which $s(t)\in \R^{2N}$ is the internal state at time $t$, $u\in \R^{2N}$ is an external input/control, $H$ is a scalar field, the coupling term $J=\begin{pmatrix}
    0 & \tilde{J}\\
    - \tilde{J} & 0
\end{pmatrix}\in \R^{2N \times 2N}$ is a skew-symmetric matrix with $\tilde{J}^\mathsf{T}=\tilde{J}\in \R^{N \times N}$, and the dissipative term $R(s,\nabla H(s))\in \R^{2N}$ is such that $\nabla R = \begin{pmatrix}
   \nabla_1 R_1 & 0\\
   0 & \nabla_2 R_2
\end{pmatrix}$, with $\nabla_1 R_1$ and $\nabla_2 R_2$ symmetric (e.g. corresponding to the diagonal matrix $\kappa$ in the dynamical equations). For these systems, one can apply our training algorithm (Scattering Backpropagation) to compute the \emph{exact gradient} since
\begin{equation}
    \nabla \mu(s^0)^{-1}=P_1(\nabla \mu(s^0)^{-1})^\mathsf{T}P_2,
\end{equation}
with $P_1=P_2=\begin{pmatrix}
    \mathbf{I}_N & 0\\
    0 & -\mathbf{I}_N
\end{pmatrix}.$
\end{eg}

\section{Further numerical simulations}

\subsubsection{Self-Kerr vs Cross-Kerr on XOR}

We trained $N=10$ (linearly) fully connected networks on XOR. For comparison, we consider both the case of self-Kerr and cross-Kerr nonlinearity, in particular, for the latter, we consider a network nonlinearly coupled in a circle, namely 
\begin{equation}
    \varphi_j(a)=a_j(|a_{j-1}|^2+|a_{j+1}|^2)
\end{equation}for every $j$ (see the schematic representation in the inset of Fig.~\ref{fig:self_vs_cross_Kerr}~\textbf{b}).
Furthermore we consider $g = 0.3$ and take $\kappa_j = 1$, $\kappa'_j = 0$ for every $j$, and we choose Xavier initialization \cite{glorot2010understanding} for the trainable linear couplings $J_{j,\ell}$ and detunings $\Delta_j$.
We encode every input $\mathrm{x} \in \mathcal{D}_\mathrm{train}=\{(0,0),(0,1),(1,0),(1,1)\}$ by choosing $\mathfrak{Re}(a_{\mathrm{in},1}) \coloneqq \mathrm{x}_1$ and $\mathfrak{Re}(a_{\mathrm{in},2}) \coloneqq \mathrm{x}_2$, and setting the other quadratures and input fields to zero. Recall that, in this Supplementary Material, we rescale the dynamical equations \eqref{eqn:system1} and work with dimension-less quantities \eqref{eq:dimension-less}.  Furthermore, we define the output of the network and the loss to be respectively $\mathrm{y} \coloneqq 10 \cdot \mathfrak{Re}(a_{\mathrm{out}5})$ and the Mean-Squared Error (MSE)
\begin{equation}
C(\mathbf{y},\mathbf{y}_\mathrm{target})=\frac{1}{4}\sum_{\mathbf{x}\in\mathcal{D}_\mathrm{train}}(\mathbf{y}(\mathbf{x})- \mathbf{y}_\mathrm{target})^2
\end{equation}
For each of $N_\mathrm{epochs}=1000$ epochs we use a RK4 method we solve the dynamical equations up to $t_\mathrm{max} = 60$ using stepsize $dt=0.1$, $\beta = 0.01$ for approximating the gradient $\partial_\theta C$ with Eqs.~\eqref{eq:dc/dl_sigma_x},~\eqref{eq:dc/dJ_sigma_x} obtained using the quasi-symmetry with $U=\sigma_x$ in Theorem \ref{theo:gradient}, and learning rate $\eta = 10^{-3}.$ As always in our numerical simulations, for every input $\mathrm{x}$ we sample a random initial condition $a(t=0)$ to solve the differential equations in the inference phase until computing $a(t_\mathrm{max})$. 
Instead, in the feedback phase, since we expect the new equilibrium to be close to the old one, we choose $a(t_\mathrm{max})$ as initial condition.
In Fig.~\ref{fig:self_vs_cross_Kerr}~\textbf{a} we plot, for ten different random initialization of the tunable parameters $\Delta_j$ and $J_{j,\ell},$ the MSE during the training of the self-Kerr network. In Fig.~\ref{fig:self_vs_cross_Kerr}~\textbf{b} we show the same data but for the network with cross-Kerr nonlinearity.

\begin{figure}[h!]
    \centering
\includegraphics[width=0.9\linewidth]{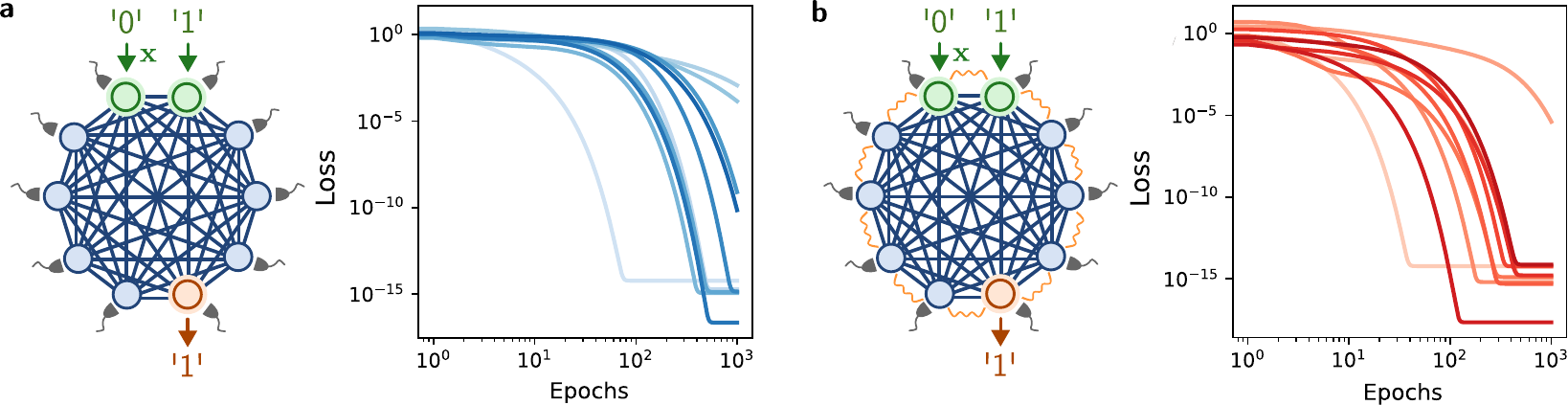}
    \caption{\textbf{a)} Training XOR with a (linearly) fully connected network of $N=10$ nodes with self-Kerr nonlinearity with strength $g=0.3$ (as sketched in the inset). Every input is encoded in the real parts of $a_{\mathrm{in},1}$ and $a_{\mathrm{in},2},$ while the output $\mathrm{y} \coloneqq 10 \cdot \mathfrak{Re}(a_{\mathrm{out},5})$. The MSE is plotted for ten different random initialization of the tunable parameters $\Delta_j$ and $J_{j,\ell}.$
    \textbf{b)} Same plot but for a network of $N=10$ nodes with cross-Kerr nonlinear coupling in a circle (see inset).
    }
    \label{fig:self_vs_cross_Kerr}
\end{figure}

In smaller size models like these we empirically observe a higher sensitivity on the random initial configuration of parameters $J,$ that we choose using Xavier initialization, meaning that for some `unlucky initial configurations' training (with respect to the same hyper-parameters) requires more epochs (as in Fig.~\ref{fig:self_vs_cross_Kerr}) or even, as it is sometimes the case for $N=3$ nodes networks, convergence is not obtained in useful times. Nevertheless, such behavior is less evident for larger networks, and not even noticeable while training MNIST on a $N=963$ node network.

\subsubsection{Robustness wrt initial conditions}
\begin{figure}[h!]
    \centering
\includegraphics[width=0.8\linewidth]{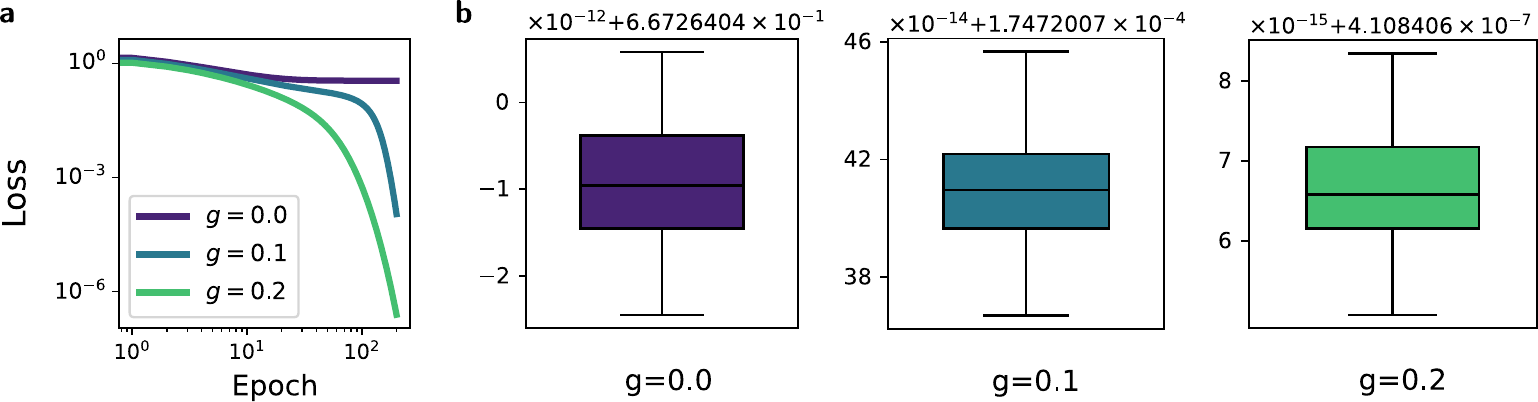}
    \caption{\textbf{a)} Mean squared error during training XOR in a $N=3$ fully connected network with self-Kerr nonlinearities of strength $g$. Models with larger values of $g$ improve faster during the $N_\mathrm{epochs}=200$ epochs, while the ablation case $(g=0.0)$ does not learn the regression task. \textbf{b)} After training, we simulate the dynamics (inference phase) starting from different $N_\mathrm{samples}=100$ random initial conditions $a(t=0).$ We show the different errors obtained in each case, highlighting almost no variability on the performance (probably due to absence of multi-stability).
    }
    \label{fig:robustness}
\end{figure}
    As discussed in the main text, we trained on XOR fully connected models of $N=3$ nodes with self-Kerr nonlinearities. In Fig.~\ref{fig:robustness}~\textbf{a} we plot the MSE loss during the $N_\mathrm{epochs}=200$ training epochs for different models, having $\kappa_j=1$ and $\kappa_j'=0$ for each $j$, but different nonlinearity strength $g$. 
    In particular, we solved the equations with RK4 for $t_\mathrm{max}=60,$ using stepsize $dt=0.01$, $\beta = \eta = 10^{-3}$ for approximating the gradient $\partial_\theta C$ with Eqs.~\eqref{eq:gradientFreuquencies},~\eqref{eq:gradientCouplings}. Input and output are encoded as above in the real parts of respectively modes $a_1$, $a_2$ and $a_3.$
    
    Furthermore, in Fig.~\ref{fig:robustness}~\textbf{b}, we show the losses obtained, for a trained model, when running the inference phase starting by a different initial configuration $a(t=0).$ Specifically, for $N_\mathrm{samples}=100$ times, we solve the dynamical equations starting from a different random initial condition: choosing $    \mathfrak{Re}(a_j(t=0)), \,\mathfrak{Im}(a_j(t=0)) \overset{\text{i.i.d.}}{\sim} \mathcal{N}(0,1)$ for each $j$.
For a fair comparison (see discussion above), we considered the same initial weights in Fig.~\ref{fig:weights}, for the training of each model.
Note that models with larger values of $g$ (up to a certain threshold when the model becomes unstable) train faster. Moreover, the fact that the loss does not vary with respect to different initial conditions $a(t=0)$ suggests that the system is not multi-stable in this parameter regime.

\begin{figure}[h!]
    \centering
\includegraphics[width=0.9\linewidth]{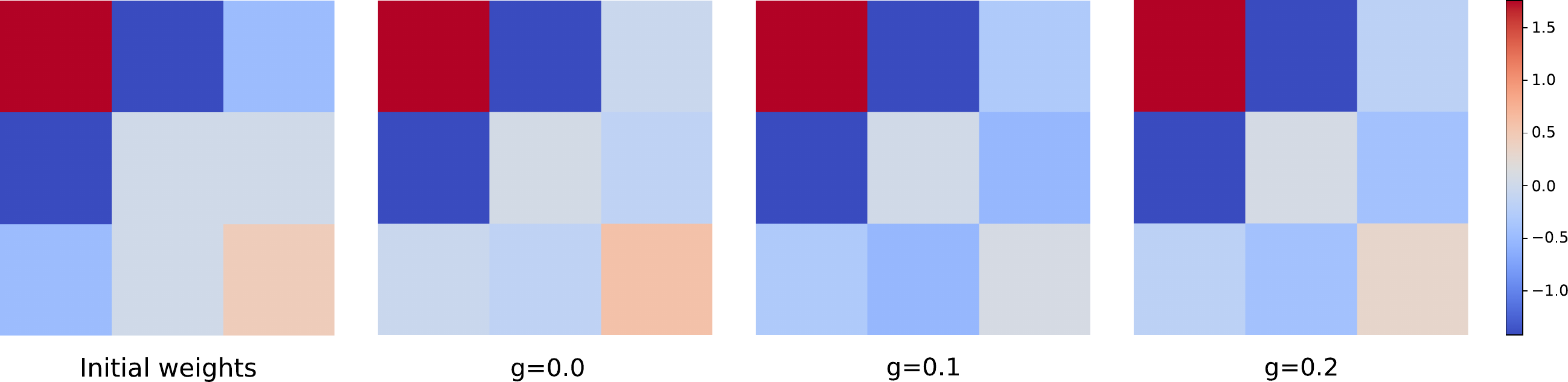}
    \caption{
    On the left, the symmetric random initial weight matrix $J$ and the correspondent trained parameters for every model. The diagonal terms represent the detuning $\Delta_j$ and the off-diagonal the couplings $J_{j,\ell}$ for $1\leq j,\ell \leq 3.$ The other plots on its right, represent the parameters after training a model (initialized with those initial weights) with self-Kerr nonlinearity with different strengths $g$.
    }
    \label{fig:weights}
\end{figure}

From Fig.~\ref{fig:weights} we can also observe how usually the learned weights are usually of the same order of their initialization, which is convenient for the physical applications. Furthermore, in this case, the learned weights in the various systems having different values of $g$ are relatively similar to each other.

\section{Unidirectional optical system trained with Scattering Backpropagation}
\begin{figure}[h!]
    \centering
\includegraphics[width=0.6\linewidth]{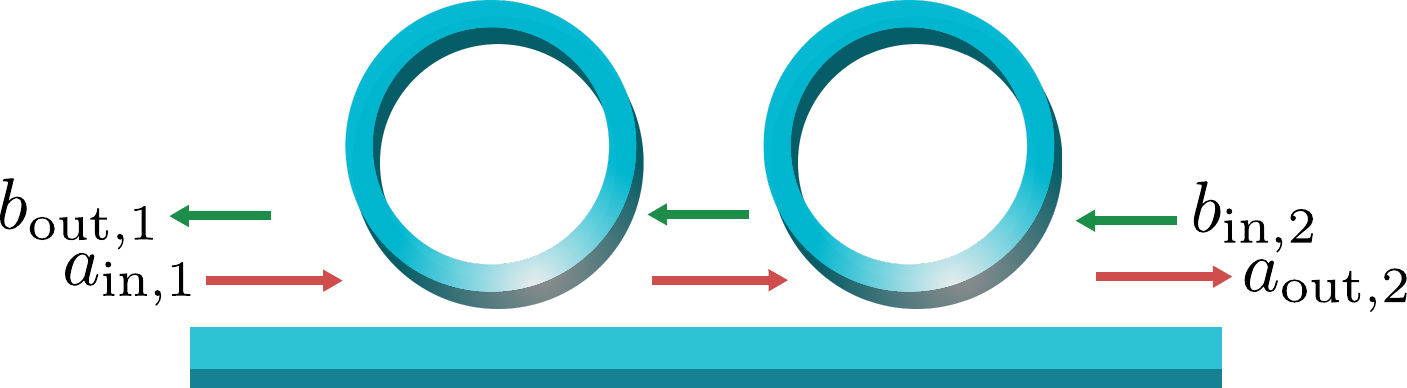}
    \caption{Unidirectional transmission in a system with two resonators coupled sequentially to a waveguide. The dynamics of the right and left-movers (respectively $a$ and $b$) can be described by the differential equation \eqref{eq:two_ring_1} and with input-output relations \eqref{eq:two_ring_inout_a_1},
    \eqref{eq:two_ring_inout_b_1}.
    Note that in the linear regime ($g=0$) the evolution of right and left-movers is decoupled, while if $g \not = 0$ this is not the case anymore. This is crucial for applying Scattering Backpropagation in this setting, encoding the input $\mathrm{x}$ into $a_{\mathrm{in},1}$ and reading the output from $a_{\mathrm{out},2},$ as it is possible to leverage the (quasi-)symmetry in the system and inject the error signal in $b_{\mathrm{in},2}.$ Remarkably, this follows directly from the general \emph{formulae} in Theorem \ref{theo:gradient} without any additional modification of the training method.
    } 
    \label{fig:two-rings}
\end{figure}
%\subsection{Quasi-symmetry and gradient estimation}
In this section, we apply our training method to a reciprocal optical scenario which however display unidirectional transport and different input-output relations. In particular, as displayed in Fig.~\ref{fig:two-rings}, we consider waveguide coupled sequentially to two ring resonators in a configuration that preserves the directionality of wave propagation. Each ring supports right-moving and left-moving modes (respectively $a_j$ and $b_j$), with dynamics influenced by Kerr nonlinearity, as described by the interaction energy
\begin{equation}
    E = g(|a|^4+|b|^4+4|a|^2|b|^2).
\end{equation}
Assuming equal coupling to the waveguide for each resonator (even though the same analysis can be done if $\kappa_1 \not=\kappa_2$), namely the dynamical equations describing the systems are
\begin{equation}
\label{eq:two_ring_1}
\begin{cases}
        \dot{a}_1=-\frac{\kappa + \kappa'_1}{2}a_1-i\Omega_1 a_1-2ig(|a_1|^2+2|b_1|^2)a_1 -\sqrt{\kappa}a_{\mathrm{in},1}\\
        \dot{a}_2=-\frac{\kappa + \kappa'_2}{2}a_2-i\Omega_2 a_2-2ig(|a_2|^2+2|b_2|^2)a_2 -\sqrt{\kappa}a_{\mathrm{in},2}\\
        \dot{b}_1=-\frac{\kappa + \kappa'_1}{2}b_1-i\Omega_1 b_1-2ig(|b_1|^2+2|a_1|^2)b_1 -\sqrt{\kappa}b_{\mathrm{in},1}\\
        \dot{b}_2=-\frac{\kappa+ \kappa'_2}{2}b_2-i\Omega_2 b_2-2ig(|b_2|^2+2|a_2|^2)b_2 -\sqrt{\kappa}b_{\mathrm{in},2}
\end{cases}
\end{equation}
together with input-output relations
\begin{equation}
\label{eq:two_ring_inout_a_1}
a_{\mathrm{out},1}=a_{\mathrm{in},1}+\sqrt{\kappa}a_1, \quad a_{\mathrm{in},2}=a_{\mathrm{out},1}, \quad
a_{\mathrm{out},2}=a_{\mathrm{in},2}+\sqrt{\kappa}a_2
\end{equation}
and
\begin{equation}
\label{eq:two_ring_inout_b_1}
b_{\mathrm{out},1}=b_{\mathrm{in},1}+\sqrt{\kappa}b_1, \quad b_{\mathrm{in},1}=b_{\mathrm{out},2}, \quad
b_{\mathrm{out},2}=b_{\mathrm{in},2}+\sqrt{\kappa}b_2.
\end{equation}
Substituting the latter into the dynamical equations \eqref{eq:two_ring_1} gives
\begin{equation}
\begin{cases}
\label{eq:two_ring_2}
        \dot{a}_1=-\frac{\kappa + \kappa'_1}{2}a_1-i\Omega_1 a_1-2ig(|a_1|^2+2|b_1|^2)a_1 -\sqrt{\kappa}a_{\mathrm{in},1}\\
        \dot{a}_2=-\kappa a_1-\frac{\kappa + \kappa'_2}{2}a_2-i\Omega_2 a_2-2ig(|a_2|^2+2|b_2|^2)a_2 -\sqrt{\kappa}a_{\mathrm{in},1}\\
        \dot{b}_1=-\kappa b_2-\frac{\kappa + \kappa'_1}{2}b_1-i\Omega_1 b_1-2ig(|b_1|^2+2|a_1|^2)b_1 -\sqrt{\kappa}b_{\mathrm{in},2}\\
        \dot{b}_2=-\frac{\kappa + \kappa'_2}{2}b_2-i\Omega_2 b_2-2ig(|b_2|^2+2|a_2|^2)b_2 -\sqrt{\kappa}b_{\mathrm{in},2}
\end{cases}
\end{equation}
Notice that the input-output relations can be simplified, in order to obtain invertible matrices $\Pi$ and $\Sigma$, by introducing new input vectors $\tilde{a}_\mathrm{in}$ and $\tilde{b}_\mathrm{in}$ by letting
\begin{equation}
    \begin{pmatrix}
        a_{\mathrm{out},1}\\
        a_{\mathrm{out},2}
    \end{pmatrix}=\begin{pmatrix}
        1 & 0 \\
        1 & 0
    \end{pmatrix}\begin{pmatrix}
        a_{\mathrm{in},1}\\
        a_{\mathrm{in},2}
    \end{pmatrix}+\begin{pmatrix}
        \sqrt{\kappa} & 0 \\
        \sqrt{\kappa} & \sqrt{\kappa}
    \end{pmatrix}\begin{pmatrix}
        a_1\\
        a_2
    \end{pmatrix}=:\begin{pmatrix}
        \tilde{a}_{\mathrm{in},1}\\
        \tilde{a}_{\mathrm{in},2}
    \end{pmatrix}+\begin{pmatrix}
        \sqrt{\kappa} & 0 \\
        \sqrt{\kappa} & \sqrt{\kappa}
    \end{pmatrix}\begin{pmatrix}
        a_1\\
        a_2
    \end{pmatrix}
\end{equation}
and 
\begin{equation}
    \begin{pmatrix}
        b_{\mathrm{out},1}\\
        b_{\mathrm{out},2}
    \end{pmatrix}=\begin{pmatrix}
        0 & 1 \\
        0 & 1
    \end{pmatrix}\begin{pmatrix}
        b_{\mathrm{in},1}\\
        b_{\mathrm{in},2}
    \end{pmatrix}+\begin{pmatrix}
        \sqrt{\kappa} & \sqrt{\kappa} \\
        0 & \sqrt{\kappa}
    \end{pmatrix}\begin{pmatrix}
        b_1\\
        b_2
    \end{pmatrix}=:\begin{pmatrix}
        \tilde{b}_{\mathrm{in},1}\\
        \tilde{b}_{\mathrm{in},2}
    \end{pmatrix}+\begin{pmatrix}
        \sqrt{\kappa} & \sqrt{\kappa} \\
         0 & \sqrt{\kappa}
    \end{pmatrix}\begin{pmatrix}
        b_1\\
        b_2
    \end{pmatrix}.
\end{equation}
After this change, the dynamical equations and input-output relations for $\xi := (a_1, a_2, b_1, b_2,a_1^*, a_2^*, b_1^*, b_2^*)^\mathsf{T}$ are respectively
\begin{equation}
    \dot{\xi} = F_\theta(\xi)-\Pi \tilde{\xi}_\mathrm{in}
\end{equation}
where
\begin{equation}
    \label{eq:two-modes-eq-xi}
\begin{cases}
        \dot{a}_1=-\frac{\kappa + \kappa'_1}{2}a_1-i\Omega_1 a_1-2ig(|a_1|^2+2|b_1|^2)a_1 -\sqrt{\kappa}\tilde{a}_{\mathrm{in},1}\\
        \dot{a}_2=-\kappa a_1-\frac{\kappa + \kappa'_2}{2}a_2-i\Omega_2 a_2-2ig(|a_2|^2+2|b_2|^2)a_2 -\sqrt{\kappa}\tilde{a}_{\mathrm{in},2}\\
        \dot{b}_1=-\kappa b_2-\frac{\kappa + \kappa'_1}{2}b_1-i\Omega_1 b_1-2ig(|b_1|^2+2|a_1|^2)b_1 -\sqrt{\kappa}\tilde{b}_{\mathrm{in},1}\\
        \dot{b}_2=-\frac{\kappa + \kappa'_2}{2}b_2-i\Omega_2 b_2-2ig(|b_2|^2+2|a_2|^2)b_2 -\sqrt{\kappa}\tilde{b}_{\mathrm{in},2}\\
        \dot{a}_1^*=-\frac{\kappa + \kappa'_1}{2}a_1^*+i\Omega_1 a_1^*+2ig(|a_1|^2+2|b_1|^2)a_1^* -\sqrt{\kappa}\tilde{a}^*_{\mathrm{in},1}\\
        \dot{a}^*_2=-\kappa a_1^*-\frac{\kappa + \kappa'_2}{2}a_2^*+i\Omega_2 a_2^*+2ig(|a_2|^2+2|b_2|^2)a_2^* -\sqrt{\kappa}\tilde{a}^*_{\mathrm{in},2}\\
        \dot{b}_1^*=-\kappa b_2^*-\frac{\kappa + \kappa'_1}{2}b_1^*+i\Omega_1 b_1^*+2ig(|b_1|^2+2|a_1|^2)b_1^* -\sqrt{\kappa}\tilde{b}^*_{\mathrm{in},1}\\
        \dot{b}^*_2=-\frac{\kappa + \kappa'_2}{2}b_2^*+i\Omega_2 b_2^*+2ig(|b_2|^2+2|a_2|^2)b_2^* -\sqrt{\kappa}\tilde{b}^*_{\mathrm{in},2}
\end{cases}
\end{equation}
and
\begin{equation}
\label{eq:in-out-new}
    \xi_\mathrm{out} = \Gamma \tilde{\xi}_\mathrm{in} + \Sigma \xi,
\end{equation}
in which $\Gamma \coloneqq \mathbf{I}_8,$ $\Pi \coloneqq \sqrt{\kappa} \,\mathbf{I}_8$, and
\begin{equation}
    \Sigma \coloneqq \begin{pmatrix}
       \sqrt{\kappa} & 0 & 0 & 0 & 0 & 0 & 0 & 0 \\
       \sqrt{\kappa} & \sqrt{\kappa} & 0 & 0 & 0 & 0 & 0 & 0\\
       0 & 0 & \sqrt{\kappa} & \sqrt{\kappa} & 0 & 0 & 0 & 0\\
       0 & 0 & 0 & \sqrt{\kappa} & 0 & 0 & 0 & 0 \\
       0 & 0 & 0 & 0 & \sqrt{\kappa} & 0 & 0 & 0\\
       0 & 0 & 0 & 0 & \sqrt{\kappa} & \sqrt{\kappa} & 0 & 0\\
       0 & 0 & 0 & 0 & 0 & 0 & \sqrt{\kappa} & \sqrt{\kappa} \\
       0 & 0 & 0 & 0 & 0 & 0 & 0 & \sqrt{\kappa}
    \end{pmatrix}.
\end{equation}
Note that the Jacobian with respect to the Wirtinger derivatives of \eqref{eq:two-modes-eq-xi} is
\begin{equation}
    DF(\bar{\xi}) = M - ig \sigma_z \frac{\partial \Phi}{\partial \xi}(\bar{\xi})
\end{equation}
where 
{\small
\begin{equation}
    M=\begin{pmatrix}
       -\frac{\kappa+\kappa'_1}{2}-i\Omega_1 & 0 & 0 & 0 & 0 & 0 & 0 & 0 \\
       -\kappa & -\frac{\kappa+\kappa'_2}{2}-i\Omega_2 & 0 & 0 & 0 & 0 & 0 & 0\\
       0 & 0 & -\frac{\kappa+\kappa'_1}{2}-i\Omega_1 & -\kappa & 0 & 0 & 0 & 0\\
       0 & 0 & 0 & -\frac{\kappa+\kappa'_2}{2}-i\Omega_2 & 0 & 0 & 0 & 0\\
       0 & 0 & 0 & 0 & -\frac{\kappa+\kappa'_1}{2}+i\Omega_1 & 0 & 0 & 0\\
       0 & 0 & 0 & 0 & -\kappa & -\frac{\kappa+\kappa'_2}{2}+i\Omega_2 & 0 & 0\\
       0 & 0 & 0 & 0 & 0 & 0 & -\frac{\kappa+\kappa'_1}{2}+i\Omega_1 & -\kappa\\
       0 & 0 & 0 & 0 &0 & 0 & 0 & -\frac{\kappa+\kappa'_2}{2}+i\Omega_2
    \end{pmatrix},
\end{equation}
}
\begin{equation}
    \sigma_z = \begin{pmatrix}
        \mathbf{I}_4 & 0\\
        0 & - \mathbf{I}_4
    \end{pmatrix},
\end{equation}
and $\frac{\partial \Phi}{\partial \xi}(\bar{\xi})$ is
{\small
\begin{equation}
     \begin{pmatrix}
       4(|\bar{a}_1|^2+|\bar{b}_1|^2) & 0 & 4\bar{a}_1\bar{b}_1^* & 0 & 2\bar{a}_1^2 & 0 & 4\bar{a}_1\bar{b}_1 & 0 \\
       0 & 4 (|\bar{a}_2|^2+|\bar{b}_2|^2)& 0 &  4\bar{a}_2\bar{b}_2^* & 0 & 2\bar{a}_2^2 & 0 & 4\bar{a}_2\bar{b}_2\\
       4\bar{a}_1^*\bar{b}_1 & 0 & 4(|\bar{a}_1|^2+|\bar{b}_1|^2) & 0 & 4\bar{a}_1\bar{b}_1 & 0 & 2 \bar{b}_1^2 & 0\\
       0 & 4\bar{a}_2^*\bar{b}_2 & 0 & 4 (|\bar{a}_2|^2+|\bar{b}_2|^2) & 0 & 4\bar{a}_2\bar{b}_2 & 0 & 2\bar{b}_2^2\\
       2(\bar{a}_1^*)^2 & 0 & 4\bar{a}^*_1\bar{b}^*_1 & 0 & 4(|\bar{a}_1|^2+|\bar{b}_1|^2) & 0 & 4\bar{a}^*_1\bar{b}_1& 0\\
       0 & 2(\bar{a}_2^*)^2 & 0 & 4\bar{a}^*_2\bar{b}^*_2 & 0 & 4(|\bar{a}_1|^2+|\bar{b}_1|^2) & 0 & 4\bar{a}^*_2\bar{b}_2\\
       4\bar{a}^*_1\bar{b}^*_1 & 0 & 2 (\bar{b}_1^*)^2 & 0 & 4\bar{a}_1\bar{b}_1^* & 0 & 4(|\bar{a}_1|^2+|\bar{b}_1|^2) & 0\\
       0 & 4\bar{a}^*_2\bar{b}^*_2 & 0 & 2(\bar{b}_2^*)^2 &0 & 4\bar{a}_2\bar{b}_2^* & 0 & 4(|\bar{a}_1|^2+|\bar{b}_1|^2)
    \end{pmatrix}
\end{equation}    
}
Notice that (see Appendix \ref{app:quasi-reciprocity})
\begin{equation}
    \frac{\partial \Phi}{\partial \xi}(\bar{\xi})=
\begin{pmatrix}
        \dfrac{\partial \varphi}{\partial (a,b)} & \dfrac{\partial \varphi}{\partial (a^*,b^*)}\\[8pt]
        \dfrac{\partial \varphi^*}{\partial (a,b)} & \dfrac{\partial \varphi^*}{\partial (a^*,b^*)}
    \end{pmatrix}
\end{equation}
where $\dfrac{\partial \varphi}{\partial (a,b)}$ and $\dfrac{\partial \varphi}{\partial (a^*,b^*)}$ are respectively an Hermitian and symmetric $4 \times 4$ matrix such that
\begin{equation}
    \dfrac{\partial \varphi}{\partial (a,b)}=\left(\dfrac{\partial \varphi^*}{\partial (a^*,b^*)}\right)^*, \quad \dfrac{\partial \varphi}{\partial (a^*,b^*)}= \left( \dfrac{\partial \varphi^*}{\partial (a,b)} \right)^*.
\end{equation}
In the linear case ($g=0$), the Jacobian it reduces to $M$ which is not symmetric. Nevertheless, by e.g. defining the matrix
\begin{equation}
    U := \begin{pmatrix}
       0 & 0 & 0 & 0 & 0 & 0 & 1 & 0 \\
       0 & 0 & 0 & 0 & 0 & 0 & 0 & 1\\
       0 & 0 & 0 & 0 & 1 & 0 & 0 & 0\\
       0 & 0 & 0 & 0 & 0 & 1 & 0 & 0 \\
       0 & 0 & 1 & 0 & 0 & 0 & 0 & 0\\
       0 & 0 & 0 & 1 & 0 & 0 & 0 & 0\\
       1 & 0 & 0 & 0 & 0 & 0 & 0 & 0 \\
       0 & 1 & 0 & 0 & 0 & 0 & 0 & 0
    \end{pmatrix}
\end{equation}
in the linear case one recovers the symmetry 
\begin{equation}
\label{eq:symm1}
    M^\dagger = U M U^{-1}.
\end{equation}
Notice that $U$ is an involutory ($U=U^{-1}$) and local transformation as it maps
\begin{equation}
    a_1 \mapsto b_1^*, \quad a_2 \mapsto b_2^*, \quad b_1 \mapsto a_1^*, \quad b_2 \mapsto a_2^*.
\end{equation}
In the nonlinear case, when $g \not = 0,$ modes $a$ and $b$ become coupled in the dynamical equations \eqref{eq:two-modes-eq-xi}; furthermore, the Jacobian's symmetry above is broken by the presence of the nonlinearity. So, informally, this model describes a ``system reaching a steady-state" which also exhibits a ``quasi-symmetry of the Green's function (or of the linearized scattering matrix)" and can be trained with Scattering Backpropagation.

In a supervised learning setting in which an input $\mathrm{x}$ is encoded into $a_{\mathrm{in},1}$ and the output $\mathrm{y}$ is decoded from $a_{\mathrm{out},2}$, the gradient approximation given by Theorem \ref{theo:gradient}, with the above $U=U_1=U_2^{-1}$, is
\begin{align}
\label{eq:grad_two_ring}
    \frac{\partial C}{\partial \Omega_1} \approx -\frac{2}{\sqrt{\kappa} \beta}\mathfrak{Im}\Bigl[ \bar{b}_1 (\delta a_{\mathrm{out},1} - \delta a_{\mathrm{in},1})+\bar{a}_1 (\delta b_{\mathrm{out},1} - \delta b_{\mathrm{in},1})-\bar{a}_1 (\delta b_{\mathrm{out},2} - \delta b_{\mathrm{in},2}) \Bigr]\\
        \frac{\partial C}{\partial \Omega_2} \approx -\frac{2}{\sqrt{\kappa} \beta}\mathfrak{Im}\Bigl[ \bar{b}_2 (\delta a_{\mathrm{out},2} - \delta a_{\mathrm{in},2})+\bar{a}_2 (\delta b_{\mathrm{out},2} - \delta b_{\mathrm{in},2}) - \bar{a}_2 (\delta b_{\mathrm{out},1} - \delta b_{\mathrm{in},1})\Bigr],
\end{align}
in which
\begin{equation}
\label{eq:err_signal_two_ring}
    \delta b_{\mathrm{in},2} := \beta \frac{\partial C}{\partial a_{\mathrm{out},2}},
\end{equation}
and the steady state components are obtained using the injected and measured fields via the input-output relations \eqref{eq:in-out-new}.

Note that, as one would physically expect, the error signal is injected at the output resonator in the left mover mode $b$. Remarkably, this follows directly from the general formula \eqref{eq:general_error_signal} with the appropriate $U,$ and no additional knowledge of the system or modification of the training method is needed for its application. 

%\subsection{Numerical experiments}

To numerically test the \emph{formulae} above we address a simple regression tasks, namely tuning the two frequencies $\Omega_1$ and $\Omega_2$ of the proposed example to reproduce the function $f(\pm 1)=\pm 1/10.$ In particular, we encode the input network $\mathrm{x}\in \{-1,1\}$ by setting $\mathfrak{Re}(a_{\mathrm{in},1})\coloneqq x$ (the other probe quadratures are zero), and define the network output as $\mathrm{y}\coloneqq \mathfrak{Re}(a_{\mathrm{out},2})$. Recalling we are working with unit-less equations, for the simulation we choose $\kappa = \kappa_1 = \kappa_2 = 1,$ $\kappa'_1=\kappa'_2 = 0$, $g=0.1$, and $\beta = \eta = 10^{-2}$. We solve the dynamical equations \eqref{eq:two_ring_2} with a Runghe-Kutta-4 method up to $t_\mathrm{max}=100$ and use \eqref{eq:grad_two_ring} for computing the approximate gradient of the mean-squared error function $C$ ---after having solved the perturbed dynamics injecting the error signal \eqref{eq:err_signal_two_ring}. As shown in Fig.~\ref{fig:two-modes-training-stats}, gradient descent with the approximation given by Scattering Propagation converges to a minimum after $100$ training epochs.

\begin{figure}
    \centering
    \includegraphics[width=0.98\linewidth]{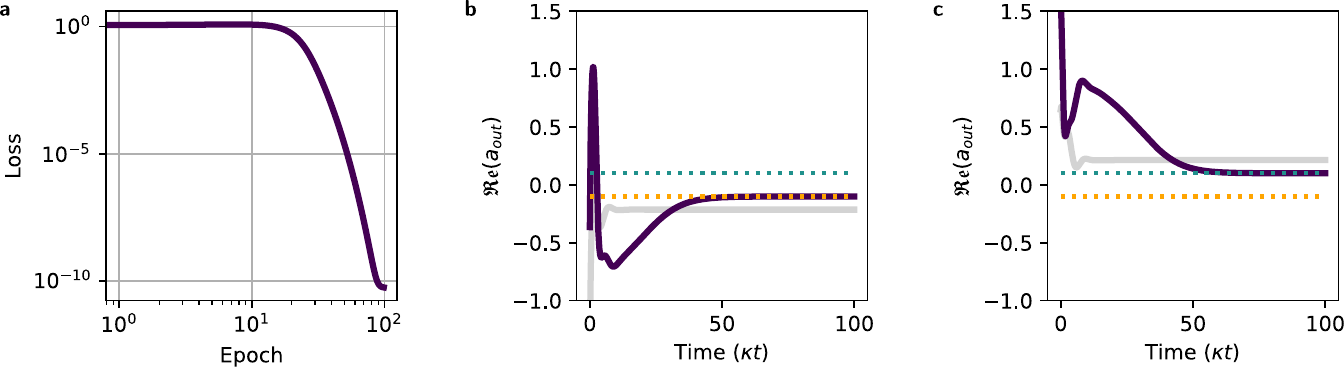}
    \caption{\textbf{a)} Mean squared error function during the $N_\mathrm{train} = 100$ training epochs. \textbf{b)} Time evolution of the trained model with $\mathfrak{Re}(a_{\mathrm{in},1})=-1$. At the steady state, the output $\mathrm{y}=\mathfrak{Re}(a_{\mathrm{out},2})$, correspondent to the blue trajectory, approaches the target $\mathrm{y_{target}}=-0.1$ (orange dotted line). \textbf{c)} Time evolution of the trained model with $\mathfrak{Re}(a_{\mathrm{in},1})=1$. In the steady state regime, $\mathfrak{Re}(a_{\mathrm{out},2})$ approaches the target $\mathrm{y_{target}}=0.1$ (green dotted line).
    }
    \label{fig:two-modes-training-stats}
\end{figure}

%TC:endignore

\end{document}